\newif\iftechrep
\techreptrue

\documentclass[letterpaper,10pt]{article}
\usepackage[section]{placeins}
\AtBeginDocument{%
  \providecommand\BibTeX{{%
    \normalfont B\kern-0.5em{\scshape i\kern-0.25em b}\kern-0.8em\TeX}}}

\usepackage{graphicx}

\usepackage{enumitem}
\usepackage{amssymb}
\usepackage{amsmath}
\usepackage[outdir=./plots/epstopdf/]{epstopdf}
\graphicspath{{./plots/}}
\usepackage{subcaption}
\usepackage{mathtools} 
\DeclarePairedDelimiter{\floor}{\lfloor}{\rfloor}
\DeclarePairedDelimiter{\ceil}{\lceil}{\rceil}
\usepackage{upgreek}
\usepackage{algorithm}
\usepackage{algorithmicx}
\usepackage{algpseudocode}
\usepackage{amsthm}
\usepackage{macros}
\usepackage{ioa_code_small-1}
\algrenewcommand\textproc{}
\PassOptionsToPackage{hyphens}{url}
\usepackage{hyperref}
\algnewcommand\algorithmicforeach{\textbf{for each}}
\algdef{S}[FOR]{ForEach}[1]{\algorithmicforeach\ #1\ \algorithmicdo}
\usepackage{xspace}
\usepackage[utf8]{inputenc}
\usepackage{textcomp}
\usepackage{booktabs}
\usepackage{tabularx}
\usepackage{multirow}
\usepackage{listings}    
\usepackage{pifont}

\newcommand{\yes}{\ding{51}}
\newcommand{\no}{\ding{55}}

\usepackage{subcaption}
\hypersetup{
  colorlinks=true,
  linkcolor=blue!50!blue, 
  urlcolor=blue!70!black,
  citecolor=blue!70!black,
}

\usepackage{ dsfont }
\usepackage{tikz}
\usetikzlibrary{shapes,arrows,positioning,automata,calc}
\tikzset{process/.style={circle,draw,inner sep=0pt, minimum size=2em}}
\pgfdeclarelayer{background}
\pgfdeclarelayer{foreground}
\pgfsetlayers{background,main,foreground}
\usetikzlibrary{external}

 \usepackage{ wasysym }
\usepackage{xcolor}
\definecolor{olive}{rgb}{0.3, 0.4, .1}
\definecolor{pinegreen}{cmyk}{0.92,0,0.59,0.25}

\usepackage{setspace}
 
\usepackage{color}
\usepackage{xcolor}
\usepackage[textwidth=0.15\textwidth,textsize=scriptsize]{todonotes}

\usepackage{amsthm}
\usepackage{eufrak}
\newtheorem{definition}{Definition}
\newtheorem{theorem}{Theorem}[section]
\newtheorem{corollary}[theorem]{Corollary}
\newtheorem{lemma}[theorem]{Lemma}

\newtheorem{innercustomgeneric}{\customgenericname}
\providecommand{\customgenericname}{}
\newcommand{\newcustomtheorem}[2]{%
  \newenvironment{#1}[1]
  {%
   \renewcommand\customgenericname{#2}%
   \renewcommand\theinnercustomgeneric{##1}%
   \innercustomgeneric
  }
  {\endinnercustomgeneric}
}

\newcustomtheorem{customthm}{Theorem}
\newcustomtheorem{customlemma}{Lemma}
\newcustomtheorem{customcor}{Corollary}

\usepackage{tikz}
\usetikzlibrary{decorations.pathreplacing,calligraphy}
\tikzstyle{cblue}=[circle, draw, thin,fill=cyan!20, scale=0.8]
\tikzstyle{qgre}=[rectangle, draw, thin,fill=green!20, scale=0.8]
\tikzstyle{rpath}=[ultra thick, opacity=0.4]
\tikzstyle{legend_isps}=[rectangle, rounded corners, thin, 
                       fill=gray!20, text=blue, draw]
                        
\tikzstyle{legend_overlay}=[rectangle, rounded corners, thin,
                           minimum width=2.5cm, minimum height=0.8cm,
                           ]
\tikzstyle{legend_phytop}=[rectangle, rounded corners, thin,
                          minimum width=2.5cm, minimum height=0.8cm,
                          ]
\tikzstyle{legend_general}=[rectangle, rounded corners, thin,
                          top color= white,bottom color=lavander!25,
                          minimum width=2.5cm, minimum height=0.8cm,
                          violet]
\pgfdeclarelayer{background}
\pgfdeclarelayer{foreground}
\pgfsetlayers{background,main,foreground}
\tikzstyle{decision} = [diamond, draw, fill=green!20, 
    text width=4.5em, text badly centered, node distance=3cm, inner sep=0pt]
\tikzstyle{block} = [rectangle, draw, fill=green!20, 
, text centered, rounded corners,align=left]
\tikzstyle{block2} = [rectangle, fill=green!20, 
, text centered, rounded corners,align=left]
\tikzstyle{line} = [draw, -latex']
\tikzstyle{cloud} = [draw, ellipse,fill=purple!20, node distance=3cm,
    minimum height=2em]

\usetikzlibrary{shapes,arrows,positioning,automata,calc,patterns}
\tikzset{node/.style={circle,draw,inner sep=0pt, minimum size=2em}}
\usetikzlibrary{external}
\usetikzlibrary{positioning}
\usetikzlibrary{fit}
\usetikzlibrary{backgrounds}
\usetikzlibrary{calc}
\usetikzlibrary{mindmap}
\usetikzlibrary{decorations.text}
\usetikzlibrary{arrows}

\tikzset{
  treenode/.style = {align=center, inner sep=0pt, text centered,
    font=\sffamily},
  arn_n/.style = {treenode, circle, white, font=\sffamily\bfseries, draw=black,
    fill=black, text width=1.5em},
  arn_r/.style = {treenode, circle, red, draw=red, 
    text width=1.5em, very thick},
  arn_x/.style = {treenode, rectangle, draw=black,
    minimum width=0.5em, minimum height=0.5em}
}
\usetikzlibrary{shapes.multipart,shapes.geometric}
\tikzset{main node/.style={circle,fill=blue!20,draw,inner sep=0pt, minimum size=2em},
}
\tikzset{node message/.style={circle split,fill=blue!20,draw,inner sep=0pt},
}
\tikzset{node home/.style={star,star points=7,star point ratio=0.8,fill=green!20,draw,minimum size=0.5cm,inner sep=0pt},}
\tikzset{man op/.style={ 
        draw,
        trapezium,
        shape border rotate=180,
        text width=2cm,
        align=center,
      },}
    
\usetikzlibrary{positioning}
\usetikzlibrary{fit}
\usetikzlibrary{backgrounds}
\usetikzlibrary{calc}
\usetikzlibrary{shapes}
\usetikzlibrary{mindmap}
\usetikzlibrary{decorations.text}
\usetikzlibrary{arrows}
\usetikzlibrary{shapes,shapes.multipart}
\usetikzlibrary{chains}
\tikzset{main node/.style={circle,fill=blue!20,draw,minimum size=1cm,inner sep=0pt},
}
\tikzset{node message/.style={circle split,fill=blue!20,draw,minimum size=0.5cm,inner sep=0pt},
}
\tikzset{node home/.style={star,star points=7,star point ratio=0.8,fill=green!20,draw,minimum size=0.5cm,inner sep=0pt},}

\usepackage{epstopdf}
\usepackage{balance}

\begin{document}
\thispagestyle{plain}
\pagestyle{plain}

\newcommand{\BVbcast}{\ensuremath{\scriptsize \lit{bv-broadcast}}}
\newcommand{\AUX}{\text{\sc aux}\xspace}
\newcommand{\ECHO}{\text{\sc echo}\xspace}
\newcommand{\POF}{\text{\sc pof}\xspace}
\newcommand{\EST}{\text{\sc est}\xspace}
\newcommand{\COORD}{\text{\sc coord}\xspace}
\newcommand{\MSG}{\text{\sc msg}\xspace}
\newcommand{\BVALECHO}{\text{\sc bvecho}\xspace}
\newcommand{\BVALREADY}{\text{\sc bvready}\xspace}
\newcommand{\INIT}{\text{\sc init}\xspace}
\newcommand{\READY}{\text{\sc ready}\xspace}
\newcommand{\CERT}{\text{\sc cert}\xspace}
\algnewcommand{\LeftComment}[1]{\Statex \(\triangleright\) #1}

\newcommand{\arcomm}[1]{\todo[color=green,bordercolor=black,linecolor=black]{\textsf{\scriptsize\linespread{1}\selectfont ARP: #1}}}
\newcommand{\arcommin}[1]{\todo[inline,color=green,bordercolor=green,linecolor=green]{\textsf{ARP: #1}}}
\newcommand{\system}{[system name] }

\newcommand{\boxedtext}[1]{\fbox{\scriptsize\bfseries\textsf{#1}}}

\newcommand{\greenremark}[2]{
   \textcolor{pinegreen}{\boxedtext{#1}
      {\small$\blacktriangleright$\emph{\textsl{#2}}$\blacktriangleleft$}
    }}

 \definecolor{burntorange}{rgb}{0.8, 0.33, 0.0}
  \newcommand{\changeremark}[2]{
   \textcolor{burntorange}{\boxedtext{#1}
      {\small$\blacktriangleright$\emph{\textsl{#2}}$\blacktriangleleft$}
}}
\newcommand{\myremark}[2]{
   \textcolor{pinegreen}{\boxedtext{#1}
      {\small$\blacktriangleright$\emph{\textsl{#2}}$\blacktriangleleft$}
    }}
  \newcommand{\myremarknew}[2]{
   \textcolor{violet}{\boxedtext{#1}
      {\small$\blacktriangleright$\emph{\textsl{#2}}$\blacktriangleleft$}
}}
\newcommand{\NewRemark}[2]{
   \textcolor{violet}{
      {\small$\blacktriangleright$\emph{\textsl{#2}}$\blacktriangleleft$}
}}

\newcommand{\redremark}[2]{
   \textcolor{red}{\boxedtext{#1}
      {\small$\blacktriangleright$\emph{\textsl{#2}}$\blacktriangleleft$}
    }}
  
  \newcommand\ARP[1]{\myremark{ARP}{#1}}
  \newcommand\ARPN[1]{\myremarknew{ARPN}{#1}}
\newcommand\vincent[1]{\redremark{VG}{#1}}
\newcommand{\warning}[1]{\redremark{\fontencoding{U}\fontfamily{futs}\selectfont\char 66\relax}{#1}}
\newcommand\NEW[1]{\NewRemark{NEW}{#1}}
\newcommand\TODO[1]{\greenremark{TODO}{#1}}
\newcommand\CHANGE[1]{\changeremark{CHANGE(?)}{#1}}


\newcommand{\cref}[1]{{\S\ref{#1}}}
\newcommand{\protocol}{Basilic\xspace} 
\newcommand{\myproperty}{active accountability\xspace} 
\newcommand{\Myproperty}{Active accountability\xspace} 
\newcommand{\mypropertyadj}{actively accountable\xspace} 
\newcommand{\Mypropertyadj}{Actively accountable\xspace} 
\newcommand{\problem}{actively accountable consensus\xspace}

\newcommand{\mypar}[1]{\paragraph{#1}}
\newcommand{\markerthree}{\raisebox{0.5pt}{\tikz{\node[scale=0.3,regular polygon, regular polygon sides=3,fill=blue!70,rotate=0](){};}}}

\newenvironment{smallenum}{
\begin{enumerate}[%
leftmargin=6pt,
labelsep=2pt,
rightmargin=0pt,
labelwidth=2pt,
itemindent=2pt,
listparindent=2pt,
topsep=4pt plus 2pt minus 4pt,
partopsep=4pt,
itemsep=0pt,
parsep=-6pt
]
    \setlength{\parskip}{-1pt}
}{\end{enumerate}}

\newenvironment{smallitem}{
\begin{itemize}[%
leftmargin=6pt,
labelsep=2pt,
rightmargin=0pt,
labelwidth=2pt,
itemindent=2pt,
listparindent=2pt,
topsep=4pt plus 2pt minus 4pt,
partopsep=4pt,
itemsep=0pt,
parsep=-6pt
]
    \setlength{\parskip}{-1pt}
}{\end{itemize}}

\date{}

\title{Basilic: Resilient Optimal Consensus Protocols With Benign and Deceitful Faults}
\author{{\rm Alejandro Ranchal-Pedrosa}\\
  \small{University of Sydney} \\
  \small{Sydney, Australia} \\
  \small{alejandro.ranchalpedrosa@sydney.edu.au}
  \and
  {\rm Vincent Gramoli}\\
  \small{University of Sydney}\\
  \small{Sydney, Australia} \\
  \small{vincent.gramoli@sydney.edu.au}}
\maketitle
\begin{abstract}
The problem of Byzantine consensus has been key to designing secure distributed systems.
However, it is particularly difficult, mainly due to the presence of Byzantine processes
that act arbitrarily and the unknown message delays in general networks.
Although it is well known that both safety and liveness are at risk as soon as $n/3$ Byzantine processes fail,
very few works attempted to characterize precisely the faults that produce safety violations 
from the faults that produce termination violations.

In this paper, we present a new lower bound on the solvability of the consensus problem by 
distinguishing deceitful faults violating safety and benign faults violating termination
from the more general Byzantine faults,
in what we call the Byzantine-deceitful-benign fault model. 
We show that one cannot solve consensus if $n\leq 3t+d+2q$ with $t$ Byzantine processes, 
$d$ deceitful processes, and $q$ benign processes.

In addition, we show that this bound is tight by presenting the Basilic class of consensus protocols
that solve consensus when $n > 3t+d+2q$.
These protocols differ in the number of processes from which they wait to receive messages 
before progressing. Each of these protocols is thus better suited for some applications depending on the 
predominance of benign or deceitful faults.

Finally, we study the fault tolerance of the Basilic class of consensus protocols 
in the context of blockchains that need to solve the weaker
problem of eventual consensus. We demonstrate that Basilic solves this problem with only $n > 2t+d+q$, 
hence demonstrating how it can strengthen blockchain security. \\

\end{abstract}
\section{Introduction}

The problem of Byzantine consensus has been key to designing secure distributed systems~\cite{singh2009zeno,CKL09,KWQ12,LVC16}. 
This problem is particularly difficult to solve because a Byzantine participant acts arbitrarily~\cite{LSP82} and message delays are generally unpredictable~\cite{DLS88}. 
Any consensus protocol would fail in this general setting if the number of Byzantine participants is $t\geq n/3$~\cite{DLS88}, where $n$ is the total number of participants.
In some executions, $\lceil n/3\rceil$ Byzantine participants can either prevent the termination of the consensus protocol by stopping or by sending unintelligible messages.
In other executions, $\lceil n/3\rceil$ can violate the agreement property of the consensus protocol by sending conflicting messages.

Interestingly, various research efforts were devoted to increase the fault tolerance of consensus protocols in closed networks (e.g., datacenters) 
by distinguishing the type of failures~\cite{CKL09,KWQ12,LVC16,LLM19}.
Some works overcome the $t<n/3$ bound by tolerating a greater number of omission than commission faults~\cite{singh2009zeno,CKL09}. 
These works are naturally well-suited for closed networks where processes are protected from intrusions by a firewall: their processes are supposedly 
more likely to crash than to be corrupted by a malicious adversary. 
In this sense, these protocols favor tolerating a greater number of faults for liveness than for safety.

Unfortunately, fewer research efforts were devoted to explore the fault tolerance of consensus protocols in open networks (e.g., blockchains).
In such settings, participants are likely to cause a disagreement if they can steal valuable assets.
This is surprising given that attacks are commonplace in blockchain systems as illustrated by 
the recent losses of $\mathdollar 70,000$\footnote{\href{https://news.bitcoin.com/bitcoin-gold-51-attacked-network-loses-70000-in-double-spends/}{https://news.bitcoin.com/bitcoin-gold-51-attacked-network-loses-70000-in-double-spends/}} and $\$18$ million\footnote{\href{https://news.bitcoin.com/bitcoin-gold-hacked-for-18-million/}{https://news.bitcoin.com/bitcoin-gold-hacked-for-18-million/}} in Bitcoin Gold, and of $\mathdollar 5.6$ million in Ethereum Classic\footnote{\href{https://news.bitcoin.com/5-6-million-stolen-as-etc-team-finally-acknowledge-the-51-attack-on-network/}{https://news.bitcoin.com/5-6-million-stolen-as-etc-team-finally-acknowledge-the-51-attack-on-network/}}.
Comparatively, some blockchain participants, called miners, are typically monitored continuously so as to ensure they provide some rewards to their owners, hence making it less likely to prevent termination.
To our knowledge, only alive-but-corrupt (abc) processes~\cite{MNR19} characterize the processes that violate consensus safety. 
Unfortunately, abc processes are restricted to only try to cause a
disagreement if the coalition size is sufficiently large to succeed at
the attempt, which is impossible to predict in blockchain systems.

\subsection{Our Results}

In this paper, we present a new lower bound on the solvability of the Byzantine consensus problem 
by precisely exploring these two additional types of faults (that either prevent termination or agreement when $t\geq n/3$).
Our lower bound states that there is no protocol solving consensus in the partially synchronous model~\cite{DLS88} if $n\leq 3t+d+2q$ with $t$ Byzantine processes, 
$d$ deceitful processes, and $q$ benign processes.
These different types of processes define the \emph{Byzantine-deceitful-benign (BDB) failure model} and are characterized by the faults they commit.
First, a \emph{deceitful} process is a process that sends some conflicting messages 
(messages that contribute to a violation of agreement) during its execution. 
Second, a \emph{benign} process is a faulty process that never sends any conflicting messages, contributing to non-termination.
For example, a benign process can crash or send stale messages, or even equivocate as long as its
messages have no effect on the agreement property. These two faults
lie at the core of the consensus problem, as the property of validity
can be locally checked for correctness by correct process, while
termination and agreement can be violated in the presence of enough
malicious processes.Compared to abc
faults, we do not impose the restriction on deceitful processes to
know whether their attack will succeed. This means that while a
protocol might tolerate $d<n/3$ abc faults along with $q<n/3$ benign
faults, it would not necessarily tolerate $d<n/3$ deceitful faults
along with $q<n/3$ benign faults. The contrary direction however
always hold.


Furthermore, we show that this lower bound is tight, in that we present the Basilic\footnote{The name ``\protocol''
is inspired from the \protocol cannon that Ottomans used to break
through the walls of Constantinople. Much like the cannon, our
\protocol protocol provides a tool to break through the classical
bounds of Byzantine fault tolerance.} 
class of protocols that solves consensus with $n> 3t+d+2q$.
Basilic builds upon recent advances in the context of accountability~\cite{civit2021} by taking into account key messages only if they are cryptographically 
signed by their sender. If they are properly signed, the recipient stores these messages and progresses in the consensus protocol execution.
Recipients also cross-check the messages they received with other recipients, based on the assumption that signatures cannot be forged.
Once conflicting messages are detected, 
they constitute an undeniable proof of fraud to exclude the faulty sender before continuing the protocol execution. 
Thanks to this exclusion, \protocol satisfies a new property, \textit{\myproperty}, which guarantees that deceitful processes can not prevent termination. 

Basilic is a class of consensus protocols, each parameterized by a different \emph{voting threshold} or the number of distinct processes from which a process 
receives messages in order to progress.
For a voting threshold of $h\in(n/2,n]$, \protocol satisfies termination if $h\leq n-q-t$, and agreement if $h>\frac{d+t+n}{2}$. This means that for just one threshold, say $h=2n/3$, \protocol tolerates multiple combinations of faulty processes: it can tolerate $t<n/3,\,q=0$ and $d=0$; but also $t=0,\,q<n/3$ and $d<n/3$; or even $t<n/6,\,q<n/6$ and $d<n/6$. This voting threshold can be modified by an application in order to tolerate any combination of $t$ Byzantine, $d$ deceitful and $b$ benign processes satisfying $n>3t+d+2q$. 
The generalization of \protocol to any voting threshold $h$ thus allows us to pick the best suited protocol depending on the application requirements. 
If, on the one hand, the application runs in a closed network (e.g., datacenter) dominated by benign processes, then the threshold will be lowered to ensure termination. 
If, on the other hand, the application runs in an open network (e.g., blockchain) dominated by deceitful processes, then the threshold will be raised to ensure agreement.

We illustrate in Figure~\ref{fig:fig1} the new resilient optimal
bounds that \protocol tolerates if there are only deceitful and benign
processes (i.e., for $t=0$), 
compared to the classic Byzantine
fault-tolerant (BFT) bound~\cite{DLS88}. We prove that these bounds
are resilient optimal in the Byzantine-deceitful-benign failure
model.
We observe that compared to state-of-the-art accountable consensus protocols,
\protocol satisfies \myproperty and tolerates a greater number of
faults, while maintaining the same time, message and bit complexities in synchronous periods.

\begin{figure}[tp]
  \center
  \includegraphics[width=.7\textwidth]{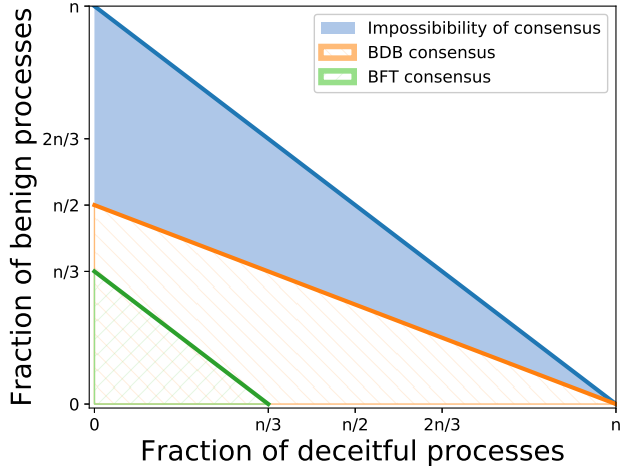}
  \caption{The green area represents the bound for BFT consensus, where $t<n/3$ and thus the same for $d,\,q$, i.e., $d+q<n/3$. The orange area is the new fault tolerance in the Byzantine-deceitful-benign (BDB) failure model, where $d<n-2q$ and $q<n/2$ (for $t=0$). In blue, the area where it is impossible to solve consensus in the BDB model.}
  \label{fig:fig1}
\end{figure}

Finally, we show that our impossibility results can be bypassed when considering a
 weaker variant of the consensus problem particularly appealing for blockchains, called the eventual consensus problem~\cite{Dubois2015} and denoted the $\Diamond$-consensus problem. We show that the \protocol class of protocols also solves $\Diamond$-consensus 
when $n>2t+d+q$, which improves our previous impossibility results by a $t$ and $q$ additive factors.
We refer to the variant of \protocol that solves both consensus and $\Diamond$-consensus as $\Diamond$-\protocol.
In particular, given a voting threshold $h$ that defines a protocol
$\sigma$ from within the $\Diamond$-\protocol class, then $\sigma$ solves
$\Diamond$-consensus if $d+t<h$ for safety and $h\leq n-q-t$ for
liveness.  


\subsection{Roadmap}
The rest of the paper is structured as follows. In
Section~\ref{sec:model}, we present the model and define the problem.
In Section~\ref{sec:imp}, we present our impossibility result in the
Byzantine-deceitful-benign model while we 
prove that \protocol protocol solves the consensus problem in Section~\ref{sec:prot} and analyze its complexities in Section~\ref{sec:comps}. In Section~\ref{sec:ec}, we introduce and prove \protocol's bounds for the eventual consensus problem.
Finally, we present the related
work in Section~\ref{sec:relw}, and we finally conclude in
Section~\ref{sec:con}.
\section{Model \& Problem}
\label{sec:model}
We consider a committee as a set $N=\{p_0, ..., p_{n-1}\}$ of $|N|=n$ processes. These processes communicate in a partially synchronous network, meaning there is a known bound $\Delta$ on the communication delay that will hold after an unknown Global Stabilization Time (GST)~\cite{DLS88}. Processes communicate through standard all-to-all reliable and authenticated communication channels~\cite{kuznesov2021}, meaning that messages can not be duplicated, forged or lost, but they can be reordered.

\paragraph{Cryptography} We assume a public-key infrastructure (PKI)
in which each party has a public key and a private key, and any
party’s public key is known to all~\cite{xue2021}. As with other protocols that use this standard assumption~\cite{xue2021,abraham2021reach}, we do not require the use of revocation lists (we will remove processes from the committee, but not from their keys from the PKI). We refer to
$\lambda$ as the security parameter, i.e., the number of bits of the
keys. As our claims and proofs require cryptography, they hold except
with $\epsilon(\lambda)$ negligible
probability~\cite{backes2003reliable}.  We formalize negligible
functions measured in the security parameter $\lambda$, which are
those functions that decrease asymptotically faster than the inverse
of any polynomial. Formally, a function $\epsilon(\kappa)$ is
negligible if for all $c>0$ there exists a $\kappa_0$ such that
$\epsilon(\kappa) < 1/\kappa^c$ for all $\kappa >
\kappa_0$~\cite{backes2003reliable}.

\paragraph{Consensus}
A protocol executed by a committee of processes solves the consensus
problem if the following three properties are satisfied by the
protocol:
\begin{itemize}
\item {\bf Termination.} Every non-faulty process eventually decides on a value.
\item {\bf Agreement.} No two non-faulty processes decide on different values.
\item {\bf Validity.} If all non-faulty processes propose the same value, no other value can be decided.
\end{itemize}

\paragraph{Conflicting messages} In order to detect faulty processes,
these have to send distinct messages to different processes where they
were expected to broadcast the same message to different processes~\cite{abraham2021}, we
refer to these messages as conflicting. Given a protocol $\sigma$, we say that a message, or set of messages, $m$ sent by process $p$ \textit{conforms} to an execution $\sigma_E$ of the protocol $\sigma$, if $\sigma_E$ belongs to the set of all possible executions where $p$ sent $m$ and $p$ is a non-faulty process. Also, a faulty process $p$ sending two messages $m,m'$ \textit{contributes} to a disagreement if there is an execution $\sigma_E$ of $\sigma$ such that (i) sufficiently many faulty processes sending $m,\,m'$ (and possibly more messages) to a disjoint subset of non-faulty processes, one to each, leads to a disagreement, and (ii) $\sigma_E$ does not lead to a disagreement without $p$ sending $m,\,m'$. Two messages $m,\, m'$ are \textit{conflicting} with respect to $\sigma$ if:
\begin{enumerate}
\item $m,\,m'$ individually conform to algorithm $\sigma$
  for some execution $\sigma_E$, $\sigma_{E'}$, respectively, $\sigma_E\neq \sigma_{E'}$,
\item there is no execution $\sigma_{E''}$ of $\sigma$
such that both messages together conform to $\sigma_{E''}$, and
\item if $p$ sending $m,m'$ to a disjoint subset of non-faulty 
processes, one to each, contributes to a disagreement.
\end{enumerate}

When combined in one message and signed by the sender, conflicting
messages constitute a proof of a process being faulty with the purpose
of causing a disagreement. We speak of this proof as a
\textit{proof-of-fraud} (PoF). An example of two conflicting messages
is a faulty process sending two different proposals for the same round
(the proposer should only propose one value per round).

Our definition of conflicting messages
differs from previous similar concepts in that conflicting messages
allow for any process $p$ to verify if two messages are conflicting: a
non-faulty process can always construct a PoF from two conflicting
messages alone, but it cannot do so with all mutant
messages~\cite{KLM03}, as $p$ would need to also learn the entire
execution, or with messages sent from an equivocating
process~\cite{CFAR12}, as these do not necessarily contribute to
disagreeing.

\paragraph{Send, receive and deliver.} Messages can
be sent and received, but we also consider broadcast primitives that
contain two functions: a broadcast function that allows process $p_i$ to
send messages to multiple channels accross the network, and a deliver
function that is invoked at the very end of the broadcast primitive to
indicate that the recipient of the message has received and processed
the message to be sent. There could be however multiple message
exchanges before the delivery can happen. As we will specify some of
these broadcast primitives, we attach the name of the protocol as a
prefix to the broadcast and deliver function to refer to a message
broadcast or delivered using that protocol, such as AARB-broadcast,
AARB-deliver, ABV-broadcast and ABV-deliver, as we detail later in
this paper.

\paragraph{Fault model} There are three mutually exclusive classes of faulty processes: Byzantine,
deceitful and benign~\cite{ranchal2020blockchain}, in what we refer to
as the \textit{Byzantine-deceitful-benign} (BDB) failure
model. Each faulty process belongs to only one of these classes. Byzantine, deceitful and benign processes are characterized by
the faults they can commit. A fault is \textit{deceitful} if it
contributes to breaking agreement, in that it sends conflicting
messages violating the protocol in order to lead two or more
partitions of processes to a disagreement. We allow deceitful
processes to constantly keep sending conflicting messages, even if
they do not succeed at causing a disagreement, but instead their
deceitful behavior prevents termination. As deceitful processes model
processes that try to break agreement, we assume also that a deceitful
fault does not send conflicting messages for rounds or phases of the
protocol that it has already terminated at the time that it sends the
messages. Deceitful processes can alternate between sending
conflicting messages and following the protocol, but cannot deviate in
any other way. A \textit{benign} fault is any fault that does not ever
send conflicting messages. Hence, benign faults cover only faults that
can break termination, e.g. by crashing, sending stale messages, etc.

As usual, Byzantine processes can act arbitrarily. Thus, Byzantine
processes can commit benign or deceitful faults, but they can also
commit faults that are neither deceitful nor benign. A fault that
sends conflicting messages and crashes afterwards is, by these
definitions, neither benign nor deceitful. We denote $t,\,d,$ and $q$
as the number of Byzantine, deceitful, and benign processes,
respectively. We assume that the
adversary is static, in that the adversary can choose up to $t$
Byzantine, $d$ deceitful and $q$ benign processes at the start of the
protocol, known only to the adversary.

In order to distinguish benign (resp. deceitful) processes
from Byzantine processes that commit a benign (resp. deceitful) fault
during a particular execution of a protocol, we formalize
fault tolerance in the BDB model. Let $E_\sigma(t,d,q)$ denote the set
of all possible executions of a protocol $\sigma$ given that there are
up to $t$ Byzantine, $d$ deceitful and $q$ benign processes. We say
that a protocol $\sigma$ for a particular problem $P$ is
\textit{$(t,d,q)$-fault-tolerant} if $\sigma$ solves $P$ for all
executions $\sigma_E\in E_\sigma(t,d,q)$. We abuse notation by
speaking of a $(t,d,q)$-fault-tolerant protocol $\sigma$ as a protocol
that tolerates $t,\,d$ and $q$ Byzantine, deceitful and benign
processes, respectively.


Note that, given a protocol $\sigma$, then $E_\sigma(0,d+k,q)\subset
E_\sigma(k,d,q)$ by definition. Thus, if $\sigma$ is
$(k,d,q)$-fault-tolerant then $\sigma$ is $(0,d+k,q)$-fault tolerant, and also $(0,d,q+k)$-fault-tolerant. However, the contrary is not
necessarily true: a protocol $\sigma$ that is
$(0,d+k,q)$-fault-tolerant is not necessarily $(k,d,q)$-fault
tolerant, as $E_\sigma(k,d,q) \nsubseteq E_\sigma(0,d+k,q)$, because
Byzantine participants can commit more faults than deceitful or
benign. Finally, a process is \textit{non-faulty} if it is neither Byzantine, nor deceitful, nor benign. Non-faulty processes follow the
protocol.

Compared to commission and omission faults, notice that not all
commission faults contribute to causing disagreements. For example,
some commission faults broadcast an invalid message that can be
discarded. In our BDB model, this type of fault would categorize as
benign, and not deceitful, since invalid messages never contribute to
a disagreement, but can instead prevent termination (by only sending
invalid messages that are discarded). All omission faults are however
benign faults, while the contrary is also not true (as per the same
aforementioned example). Compared to the alive-but-corrupt failure
model, deceitful faults are not restricted to only contribute to a
disagreement if they know the disagreement will succeed, but instead 
we let them try forever, even if they do not succeed. Also, the
alive-but-corrupt failure model does not define benign faults.

We believe thus the BDB model to be better-suited for consensus, as it
establishes a clear difference in the types of faults depending on the
type of property that the fault jeopardizes (agreement for deceitful,
termination for benign), without restricting the behavior of these faults
to the cases where they are certain that they will cause a disagreement.


\section{Impossibility Results}
\label{sec:imp}

In this section, we extend Dwork et al.'s impossibility results~\cite{DLS88} on the
number of processes necessary to solve the Byzantine
consensus problem with partial synchrony by adding deceitful and benign processes. First, we prove
in Section~\ref{sec:impgen} lower bounds on the size of the
committee of any consensus protocol. Then,  we prove in
Section~\ref{sec:impthre} some lower bounds depending on the
voting threshold of that protocol, which we define in the same
section.

\subsection{Impossibility of consensus in the BDB model}
\label{sec:impgen}
First, we consider the case where $t=0$, i.e., there are only
deceitful and benign processes. In particular, we show in
Lemma~\ref{lem:imp} that if a protocol solves consensus then it
tolerates at most $d<n-2q$ deceitful processes and $q< n/2$ benign
processes. The intuition for the proof is the same from the classical
impossibility proof of consensus in partial synchrony in the presence
of $t_0+1$ Byzantine processes. Lemma~\ref{lem:imp} extends the the
BDB model the classical lower bound for the BFT model~\cite{DLS88}, by
tolerating a stronger adversary than the classical bound (e.g. an
adversary causing $d=\ceil{n/3}-1$ deceitful faults and
$q=\ceil{n/3}-1$ benign faults). By contradiction, we show that in the
presence of a greater number of faulty processes than bounded by
Lemma~\ref{lem:imp}, in some executions all processes would either not
terminate, or not satisfy agreement, if maintaining validity.

\begin{lemma}
  \label{lem:imp}
  Let a protocol $\sigma$ and let $\sigma$ solve consensus for all executions $\sigma_E\in E_\sigma(0,d,q)$ for some $d,q>0$. Then, $d< n-2q$ and $q<n/2$.
\end{lemma}
\begin{proof}

  First, we show $q<n/2$ by contradiction, as done by previous work
for omission faults~\cite{DLS88}. Suppose $q\geq n/2,\,d=0$ and
consider processes are divided into a disjoint partition $P,Q$ such
that $P$ contains between $1$ and $q$ processes and $Q$ contains
$n-|P|$. First, consider scenario A: all processes in $P$ are benign
and the rest correct, and all processes in $Q$ propose value
$0$. Then, by validity all processes in $Q$ decide $0$. Then, consider
scenario B: all processes in $Q$ are benign and the rest correct, and
all processes in $P$ propose value $1$. Then, by validity all
processes in $P$ decide $1$. Now consider scenario C: no process is
benign, and processes in $P$ propose all $0$ while processes in $Q$
propose all value $1$. For processes in $P$ scenario C is
indistinguishable from scenario A, while for processes in $Q$ scenario
B is indistinguishable from scenario A. This yields a
contradiction.

It follows that $q<n/2$. Hence, for $n=2$, and since $q<1$, it is
immediate that for $d\geq 2$ it is impossible to solve consensus. As
such, we have left to consider $d\geq n-2q$ with $n\geq 3$. We will
prove this by contradiction.

Consider processes are divided into three disjoint partitions $P, Q, R$, such
that $P$ and $Q$ contain between $1$ and $q$ processes each, and $R$
contains between $1$ and $d$. First consider the following scenario A:
processes in $P$ and $R$ are non-faulty and propose value $0$, and
processes in $Q$ are benign. It follows that $P\cup R$ must decide
value $0$ at some time $T_A$, for if they decided $1$ there would be a
scenario in which processes in $Q$ are non-faulty and also propose $0$,
but messages sent from processes in $Q$ are delivered at a time
greater than $T_A$, having processes in $P\cup R$ already decided
$1$. This would break the validity property. Also, they must decide
some value to satisfy termination tolerating $q$ benign faults.

Consider now scenario B: processes in $P$ are benign, and processes in
$R$ and $Q$ are non-faulty and propose value $1$. By the same approach,
$R\cup Q$ decide $1$ at a time $T_B$.

Now consider scenario C: processes in $P$ and $Q$ are non-faulty, and
processes in $R$ are deceitful, the messages sent from
processes in $Q$ are delivered by processes in $P$ at a time greater
than $\max(T_A,T_B)$, and the same for messages sent from processes in
$P$ to processes in $Q$. Then, for processes in $P$ this scenario is
identical to scenario A, deciding $0$, while for processes in $Q$ this
is identical to scenario B, deciding $1$, which leads to a
disagreement. This yields a contradiction.
\end{proof}

\begin{corollary}[Impossibility of consensus with $t=0$]
  \label{cor:imp}
  It is impossible for a consensus protocol $\sigma$ to tolerate $d$ deceitful and $q$ benign processes if $d\geq n-2q$ or $q\geq n/2$.
\end{corollary}
\begin{proof}
  This is immediate from Lemma~\ref{lem:imp} since $\sigma$ is $(0,d,q)$-fault-tolerant if $\sigma$ solves $P$ for all
executions $\sigma_E\in E_\sigma(0,d,q)$. 
\end{proof}

We prove the impossibility result of Theorem~\ref{thm:imp} by extending the result of Corollary~\ref{cor:imp}: it is
impossible to solve consensus in the presence of $t$ Byzantine, $q$
benign and $d$ deceitful processes unless $n>3t+d+2q$.
\begin{theorem}[Impossibility of consensus]
  \label{thm:imp}
      It is impossible for a consensus protocol to tolerate $t$ Byzantine, $d$ deceitful and $q$ benign processes if $n\leq 3t+d+2q$.
\end{theorem}
\begin{proof}
  The proof is analogous to that of Lemma~\ref{lem:imp} since $\sigma$
is $(t,d,q)$-fault-tolerant if $\sigma$ solves $P$ for all executions
$\sigma_E\in E_\sigma(t,d,q)$, and $E_\sigma(0,d+t,q)\in E_\sigma(t,d,q)$
and $E_\sigma(0,d,q+t)\in E_\sigma(t,d,q)$ by definition. Thus, the bounds
in this case become $d+t\geq n-2(q+t)$ and $q+t\geq n/2$, which
results in $n+t+d< 2n-2q-2t\iff n>3t+d+2q$.
\end{proof}

\subsection{Impossibility bounds per voting threshold}
\label{sec:impthre}
The proofs for the impossibility results of Section~\ref{sec:impgen}
(and for the classical impossibility results~\cite{DLS88}) derive a
trade-off between agreement and termination. In some scenarios,
processes must be able to terminate without delivering messages from a
number of processes that may commit benign faults. In other scenarios,
processes must be able to deliver messages from enough processes
before terminating in order to make sure that no disagreement caused
by deceitful faults is possible. We prove in this section the
impossibility results depending on this trade-off.

A protocol that satisfies both
agreement and termination in partial synchrony must thus state a
threshold that represents the number of processes from which to
deliver messages in order to be able to terminate without compromising
agreement. If this threshold
is either too small to satisfy
agreement, or too large to satisfy termination, then the protocol does
not solve consensus. We refer to this threshold as the \textit{voting
threshold}, and denote it with $h$. Typically, this threshold is
$h=n-t_0= \ceil{\frac{2n}{3}}$ to tolerate $t_0=\ceil{\frac{n}{3}}-1$
Byzantine faults~\cite{crain2018dbft,civit2021, KADC07, YMR19,
ranchal2020blockchain}. We prove however in Lemma~\ref{lem:impagr} and
Corollary~\ref{cor:impagr} that $h>\frac{d+t+n}{2}$ with $h\in(n/2,n]$
for safety.

\begin{lemma}[Impossibility of Agreement ($t=0$)]
  \label{lem:impagr}
  Let $\sigma$ be a protocol with voting threshold $h\in(n/2,n]$ that satisfies agreement. Then $\sigma$ tolerates at most $d<2h-n$ deceitful processes.
\end{lemma}
\begin{proof}
    The bound $h\in(n/2,n]$ derives trivially: if $h\leq n/2$ then two
subsets without any faulty processes can reach the threshold for
different values (Lemma~\ref{lem:imp}).
  We calculate for which cases it is possible to cause a
disagreement. Hence, we have two disjoint partitions of non-faulty processes such
that $|A|+|B|\leq n-d$. Suppose that processes in $A$ and in $B$ decide each a
different decision $v_A,\,v_B,\, v_A\neq v_B$. This means that both $|A|+d\geq h$ and
$|B|+d\geq h$ must hold. Adding them up, we have
$|A|+|B|+2d\geq 2h$ and since $|A|+|B|\leq n-d$ we have
$n+d\geq 2h$ for a disagreement to occur. This means that
if $h>\frac{n+d}{2}$ then it is impossible for $d$ deceitful
processes to cause a disagreement.
\end{proof}
The proof of Lemma~\ref{lem:impagr} can be straightforwardly extended to
include Byzantine processes, resulting in
Corollary~\ref{cor:impagr}. 
\begin{corollary}
  \label{cor:impagr}
  Let $\sigma$ be a protocol with voting threshold $h\in(n/2,n]$ that satisfies agreement. Then $\sigma$ tolerates at most $d+t<2h-n$ deceitful and Byzantine processes.
\end{corollary}

Next, in Lemma~\ref{lem:impter} and Corollary~\ref{cor:impter} we show the analogous results for
the termination property. That is, we show that if a protocol solves
termination while $t=0$, then it tolerates at most $q\leq n-h$ benign
processes, or $q+t\leq n-h$ benign and Byzantine processes.

\begin{lemma}[Impossibility of Termination ($t=0$)]
  \label{lem:impter}
  Let $\sigma$ be a protocol with voting threshold $h$ that satisfies termination. Then $\sigma$ tolerates at most $q\leq n-h$ benign processes.
\end{lemma}
\begin{proof}
  If $n-q$ benign processes are less than $h$, then termination is not
  guaranteed, since in this case termination would require the votes from some benign processes. This is impossible if
  $h\leq n-q$, as it guarantees that the threshold is lower than
  all processes minus the
  $q$ benign processes.
\end{proof}

\begin{corollary}
  \label{cor:impter}
  Let $\sigma$ be a protocol with voting threshold $h$ that satisfies termination. Then, $\sigma$ tolerates at most $q+t\leq n-h$ benign and Byzantine processes. 
\end{corollary}
Combining the results of corollaries~\ref{cor:impagr} and~\ref{cor:impter}, one can derive an impossibility bound for a consensus protocol given its voting threshold. We show this result in Corollary~\ref{cor:impteragr}.
\begin{corollary}
  \label{cor:impteragr}
  Let $\sigma$ be a protocol that solves the consensus problem with voting threshold $h\in(n/2,n]$. Then, $\sigma$ tolerates at most $d+t<2h-n$ and $q+t\leq n-h$ Byzantine, deceitful and benign processes. 
\end{corollary}
We show in Figure~\ref{fig:fig2} the threshold $h$ to tolerate a
number $d$ of deceitful and $q$ of benign processes. For example, for a
threshold $h=\ceil{\frac{5n}{9}}-1$, we have that $d<\frac{n}{9}$ for
safety and $q<\frac{4n}{9}$ for liveness, with $t=0$. The maximum
number of Byzantine processes tolerated with $d=q=0$ is the minimum of
both bounds, being for example $t<\frac{n}{9}$ for
$h=\ceil{\frac{5n}{9}}-1$. In the remainder of this paper, we assume
the adversary satisfies the resilient optimal bounds of $h<n-q-t$ and $h>\frac{d+t+n}{2}$, given a
particular voting threshold $h$. The result of Theorem~\ref{thm:imp}
holds regardless of the voting threshold. Thus, a protocol that
satisfies both $h<n-q-t$ and $h>\frac{d+t+n}{2}$ can set its voting
threshold $h\in(n/2,n]$ in order to solve consensus for any
combination of $t$ Byzantine, $q$ benign and $d$ deceitful processes, as long as $n>3t+d+2q$ holds.

\begin{figure}[ht] \center
\includegraphics[width=.7\textwidth]{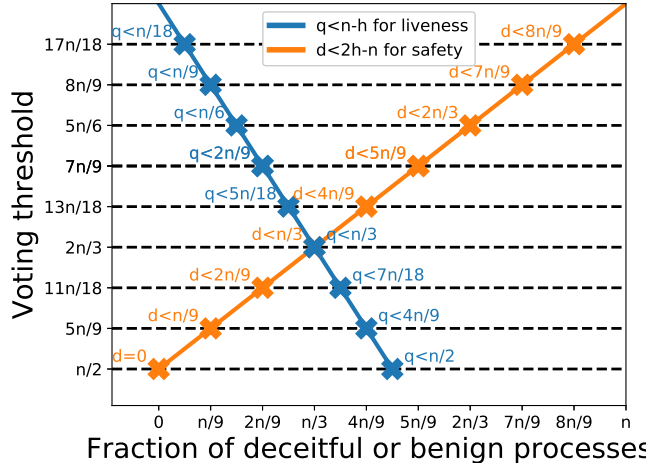}
  \caption{Number of
deceitful processes $d$ and benign processes $q$ tolerated for safety and liveness, respectively, per voting threshold $h$ and with $t=0$ Byzantine processes.}
  \label{fig:fig2}
\end{figure}
  \section{The \protocol Protocol}
  \label{sec:prot}
  In this section, we introduce the \protocol class of protocols, 
  a class of resilient optimal protocols that solve, for different voting thresholds, the \problem problem in the BDB model. 
  In particular, all protocols within the \protocol class tolerate $t$ Byzantine, $d$ deceitful and $q$ benign processes satisfying $n>3t+d+2q$, and, given a particular protocol $\sigma(h)$ of the class uniquely defined by a voting threshold $h\in(n/2,n]$, then $\sigma(h)$ tolerates a number $n$ of processes satisfying $d+t<2h-n$ and $q+t\leq n-h$. In this section, we first need to introduce few assumptions and definitions in Section~\ref{sec:model2}. Second, we present the overview of the \protocol protocol in Section~\ref{sec:overview}, and show its components in sections~\ref{sec:aabc},~\ref{sec:aarb}, and~\ref{sec:gen}.
    \subsection{Additional Assumptions}
    \label{sec:model2}

  \paragraph{Adversary} In order to limit the computational power of processes to prevent the adversary from forging keys, we model processes as probabilistic
polynomial-time interactive Turing machines
(ITMs)~\cite{Dumbo2020,cachin2005random,cachin2001}. A process is an
ITM defined by the following protocol: it is activated upon receiving
an incoming message to carry out some computations, update its states,
possibly generate some outgoing messages, and wait for the next
activation. The adversary $\mathcal{A}$ is a probabilistic ITM that
runs in polynomial time (in the number of message bits generated by
non-faulty processes). 

\paragraph{{\expandafter\MakeUppercase\problem} problem}
The accountable consensus problem~\cite{civit2021} includes the property of
accountability in order to provide guarantees in the event that
deceitful and Byzantine processes manage to cause a disagreement. This
property is however insufficient for the purpose of \protocol. We need
an additional property that identifies and removes all deceitful
behavior that prevents termination. Faulty processes can break agreement
in a finite number of conflicting messages, but once they send a pair
of these conflicting messages, they leave a trace that can result in
their exclusion from the system. Our goal is to exploit this trace to make sure that
deceitful processes cannot contribute to breaking liveness. As a result,
we include the property of \myproperty, stating that deceitful faults do not
prevent termination of the protocol.

\begin{definition}[{\expandafter\MakeUppercase\problem} problem]
A protocol $\sigma$ with voting threshold $h$ solves the \problem
problem if the following properties are satisfied:
\begin{itemize}
\item {\bf Termination.} Every non-faulty process eventually decides on a value.
\item {\bf Validity.} If all non-faulty processes propose the same value, no other value can be decided.
\item {\bf Agreement.} If $d+t<2h-n$ then no two non-faulty processes decide on different values.
\item {\bf Accountability.} If two non-faulty processes output
disagreeing decision values, then all non-faulty processes eventually
identify at least $2h-n$ faulty processes responsible for that
disagreement.
\item{\bf \expandafter\MakeUppercase\myproperty.} Deceitful behavior
  does not prevent liveness.
\end{itemize}
\end{definition}
We generalise the previous definition of accountability~\cite{civit2021} by including the voting threshold $h$. That is, the previous definition of accountability is the one we present in this work for the standard voting threshold of $h=2n/3$.



\subsection{Basilic Internals}
\label{sec:overview}

\protocol is a class of consensus protocols, all these protocols follow the same 
pseudocode (Algorithms~\ref{alg:prot}--\ref{alg:gen}) but differ by their voting threshold $h\in(n/2,n]$.
The structures of these protocols follow the classic reduction~\cite{BCG93} from the consensus problem, 
which accepts any ordered set of input values, to the binary consensus problem, which accepts binary
input values.

\paragraph{Basilic Overview}
More specifically, Basilic has at its core the binary consensus
protocol called \textit{\mypropertyadj binary consensus} or AABC for
short (Alg.~\ref{alg:prot}--\ref{alg:helper}) and presented in
Section~\ref{sec:aabc}. We show in Figure~\ref{fig:example} an example
execution with $n=4$ processes in the committee. First each process
$p_i$ selects their input value $v_i$, which they share with everyone
executing an instance of a reliable broadcast protocol called
\textit{\mypropertyadj reliable broadcast} or AARB for short. Then,
processes executed one instance $AABC_i$ of the binary consensus
protocol to decide whether to select their associated input value from
process $p_i$. Finally, processes locally process the minimum input
value from the values whose associated AABC instance output $1$.

This Basilic binary consensus protocol shares similarities with Polygraph~\cite{CGG21}, as it also detects
guilty processes, but goes further, by excluding these detected processes and adjusting its 
voting threshold at runtime to solve consensus even in cases where Polygraph cannot ($ n/3\leq t+q+d < n$). We summarize the comparison of \protocol with the state of the art in Table~\ref{tab:bigtab}.
Finally, the rest of the reduction is depicted in Alg.~\ref{alg:gen} and invokes $n$ \mypropertyadj reliable 
broadcast instances or AARB (Alg.~\ref{alg:arb}) and described in Section~\ref{sec:aarb}, 
followed by $n$ of the aforementioned AABC instances.

\tikzstyle{n}= [circle, fill, minimum size=10pt,inner sep=0pt, outer
sep=0pt] \tikzstyle{mul} = [circle,draw,inner sep=-1pt]
\newcounter{y}
\begin{figure*}[h]
  \hspace{-5em}
  \begin{tikzpicture}[yscale=0.5, xscale=1.2, node distance=0.3cm,
    auto]


    \def\varn{3}
    \def\varz{-1}
    \def\varf{0}
    \def\vars{2.75}
    \def\vart{5.75}
    \def\varmv{7}
    \def\varfo{9}
    \def\varfif{12}
    \foreach \i in {0,...,\varn}{
      \node (0-\i) at (\varf,-\i) {$p_\i: v_\i$};
      \node (1-\i) at (\vars,-\i) {AARB$_\i:v_\i$};
    }

    \draw [decorate, 
    decoration = {calligraphic brace,
        raise=5pt,
        amplitude=5pt}] (0-0.north)+(-0.25,0) --  (1-0.north)+(-0.25,0)
      node[pos=0.5,above=10pt,black]{reliably broadcast proposals};
      
    \foreach \i/\y in {0/1,1/0,2/1,3/0}{
      \node (2-\i) at (\vart,-\i) {AABC$_\i:\y$};
    }

    \node (3-i) at (\varfo,-1.5) {\scriptsize $\{v_0:1,\,v_1:0,\,v_2:1,\,v_3:0\}$};
    \node (4-i) at (\varfif,-1.5){$v_0$};
    
    \draw [decorate, 
    decoration = {calligraphic brace,
        raise=5pt,
        amplitude=5pt}] (1-0.north)+(0.25,0) -- (2-0.north)
      node[pos=0.6,above=10pt,black]{binary consensus decisions};
      
    \node[] (MV) at (\varmv,-1.5) {};
    \draw[-latex'] (3-i.east) -- node[above, midway]{\scriptsize $min(v_0,v_2)$} (4-i.west);
    \draw[-latex'] (MV.center) -- (3-i.west);

    \draw [decorate, 
    decoration = {calligraphic brace,
        raise=5pt,
        amplitude=5pt}] (2-0.north)+(0.25,0) -- (10.4,0.5)
      node[pos=0.5,above=10pt,black]{bitmask and associated bits};

      \draw [decorate, 
    decoration = {calligraphic brace,
        raise=5pt,
        amplitude=5pt}] (10.5,0.5) -- (12.2,0.5)
      node[pos=0.5,above=10pt,black]{decide one};

      \draw [decorate, 
    decoration = {calligraphic brace,
        raise=5pt,
        amplitude=5pt}] (-0.25,2) -- (12.2,2)
      node[pos=0.5,above=10pt,black]{\protocol's multi-valued consensus};
    
    \foreach \i in {0,...,\varn}{
      \draw (2-\i.east) -| (MV.center);        
      
      \foreach \j in {0,...,\varn}{
        \draw (0-\i.east) -- (1-\j.west);
        \draw (1-\i.east) -- (2-\j.west);

        
      }
    }

  \end{tikzpicture}
  \caption{\protocol execution example for a committee of
$n=4$. First, each process $p_i$ selects their input value $v_i$,
which they share with everyone executing their respective instance
$AARB_i$ of $AARB$. Then, processes executed one instance $AABC_i$ of
the binary consensus protocol to decide whether to select their
associated input value from process $p_i$. Finally, processes locally
process the minimum input value from the values whose associated AABC
instance output $1$.}
  \label{fig:example}
\end{figure*}
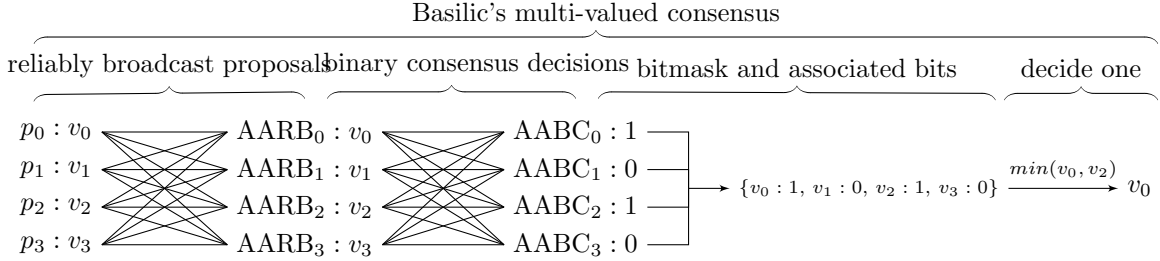


\mypar{Certificates and transferable authentication} \protocol uses certificates in order to validate or discard a message, and also to detect deceitful processes
by cross-checking certificates. A certificate is a list of previously
delivered and signed messages 
that justifies the content of the message in which the certificate is piggybacked.
Thus, non-faulty processes perform
transferable authentication~\cite{CFAR12}. That is, process $p_i$ can
deliver $\ms{msg}$ from $p_j$ by verifying the signature of
$\ms{msg}$, even if $\ms{msg}$ was received from $p_k$, for
$k\neq i\neq j$.

\mypar{Detected deceitful processes} A key novelty of \protocol is to
remove detected deceitful processes from the committee at
runtime. For this reason, we refer to $d_r$ as the number of detected
deceitful processes, and define a voting threshold $h(d_r)$ that
varies with the number of detected deceitful processes. Therefore, processes start \protocol with an
initial voting threshold $h(d_r=0)=h_0$, e.g., $h_0=\ceil{\frac{2n}{3}}$, but then
update the threshold by removing detected deceitful processes,
i.e. $h(d_r)=h_0-d_r$. This way, detected deceitful processes break
neither liveness nor safety, as we will show. Certificates must always
contain $h(d_r)$ signatures from
distinct processes justifying the message (after filtering out up to
$d_r$ signatures from detected deceitful processes), or else they will
be discarded. Recall that the adversary is thus constrained to the bounds from Corollary~\ref{cor:impteragr} depending on the voting threshold. As \protocol uses a threshold that updates at runtime starting from an initial threshold $h(d_r)=h_0-d_r$, we restate these bounds 
applied to the initial threshold $h_0\leq n-q-t$ and $h_0>\frac{d+t+n}{2}$, or to the updated threshold of $h(d_r)<n-q-t-d_r$ and
$h(d_r)>\frac{d+t+n}{2}-d_r$.


\subsection{The General \protocol Protocol}
\label{sec:gen}
We bring together the $n$ instances of the AABC binary consensus
protocol with the $n$ instances of the AARB reliable broadcast protocol in
Algorithm~\ref{alg:gen}, where we show the general \protocol
protocol. The protocol derives from Polygraph's general protocol~\cite{civit2019techrep,civit2021},
which in turn derives from DBFT's multi-valued consensus protocol~\cite{crain2018dbft}.

Non-faulty processes first start the AARB protocol for which they 
are the source by proposing a value in line~\ref{line:genaarb}. Delivered proposals are stored in an array $\ms{proposals}$ at the index
 corresponding to the source of the proposal. A binary consensus
at index $k$ is started with input value 1 for each index $k$ where a
proposal has been recorded (line~\ref{line:genaabc}). Notice
that we can guarantee to decide $1$ on at most $h(d_r)$ proposals
(line~\ref{line:gendec1}), where $d_r$ can be up to $d$, meaning that,
for the standard threshold $h(d_r)=\ceil{\frac{2n}{3}}-d_r$, the
maximum number of decided proposals is $\ceil{\frac{n}{3}}$, since
$d_r<\frac{n}{3}$. Once non-faulty processes decide 1 on at lest $h(d_r)$ AABC instances, non-faulty processes start the remaining AABC
instances with input value 0 (line~\ref{line:genaabc0}), without
having to wait to AARB-deliver their respective values.

Finally, once all AABC instances have terminated
(line~\ref{line:genaabcter}), non-faulty processes can output a decision. As such, processes take as input a list of AARB-delivered values and their associated index and output a decision selecting the AARB-delivered value with the lowest associated index whose binary consensus with the same index output 1 (line~\ref{line:genuni}).


\begin{algorithm}[th]
  \caption{The general \protocol with initial threshold $h_0$.}
  \label{alg:gen}
  \smallsize{
    \begin{algorithmic}[1]
      \Part{\smallsize $\lit{\protocol-gen-propose_{h_0}(}v_i\lit{)}$}{
        \State $\ms{msgs}\gets \lit{AARB-broadcast(}\EST, \langle v_i,i\rangle)$
        \Comment{Algorithm~\ref{alg:arb}}\label{line:genaarb}
        \Repeat{}
        \SmallIf{$\exists v,k:(\EST,\langle v,k\rangle)\in\ms{msgs}$}{}
        \Comment{proposal AARB-delivered}
        \SmallIf{$\lit{BIN-CONSENSUS}[k]$ not yet invoked}{}\Comment{Algorithm ~\ref{alg:prot}}
        \State \hspace{-0.5em}$\ms{bin-decisions}[k]\gets\lit{BIN-CONSENSUS}[k].\lit{AABC-prop}(1)$\label{line:genaabc}
        \EndSmallIf
        \EndSmallIf
        \EndRepeat
        \Until{$|\ms{bin-decisions}[k]=1|\geq h(d_r)$} \Comment{decide $1$ on at least $h(d_r)$}\label{line:gendec1}\EndUntil
        \For{\textbf{all} $k$ such that $\lit{BIN-CONSENSUS}[k]$ not yet invoked}

        \State $\ms{bin-decisions}[k]\gets\lit{BIN-CONSENSUS}[k].\lit{AABC-prop}(0)$

        \label{line:genaabc0}
        \EndFor
        \WUntil{\textbf{for all} $k,\,\ms{bin-decisions}[k]\neq \bot$}\EndWUntil\label{line:genaabcter}
            
        \State $j\gets \min\{k:\ms{bin-decisions}[k]=1\} \lit{)} $\label{line:genuni1}        
        \WUntil{$\exists v:(\EST, \langle v,j\rangle)\in\ms{msgs}$}\EndWUntil
        \State \textbf{decide} $v$\label{line:genuni}        
      }\EndPart
      \algstore{alg:gen}
    \end{algorithmic}
  }
\end{algorithm}

\subsection{\Mypropertyadj Binary Consensus}
\label{sec:aabc}
We show in Algorithm~\ref{alg:prot} the \protocol
\textit{\mypropertyadj binary consensus} (AABC) protocol with initial threshold $h_0\in(n/2,n]$, along with some additional components and functions in Algorithm~\ref{alg:helper}. First, note that all delivered messages are correctly signed 
(as wrongly signed messages are discarded)
and
stored in $\ms{sig\_msgs}$, along with all sent messages (as we
detail in Rule~\ref{item:signed} of Alg.~\ref{alg:prot}).

\begin{algorithm}[htp]
  \caption{\protocol's AABC with initial threshold $h_0$ for $p_i$.
  }
  \label{alg:prot}
  \smallsize{
    \begin{algorithmic}[1]
      \algrestore{alg:gen}
      \Part{\smallsize $\lit{AABC-prop_{h_0}(}v_i\lit{)}$}{
        \State $\ms{est}\gets v_i$\label{line:estimate}
        \State $\ms{r}\gets 0$
        \State $\ms{timeout}\gets 0$
        \State $\ms{cert}[0]\gets \emptyset$
        \State $\ms{bin\_vals} \gets \emptyset$
        \Repeat{}
        \State $r\gets r+1$
        \State $\ms{timeout}\gets \Delta$\Comment{set timer}
        \State $coord\gets ((r-1)mod\,n)+1$\Comment{rotate coordinator}\label{line:rotcoord}
        \EndRepeat
        \Part{$\blacktriangleright$ Phase 1}{
          \State $timer\gets \lit{start-timer}(\ms{timeout})$ \Comment{start timer}
          \State $\lit{abv-broadcast}(\EST[r],\ms{est},\ms{cert}[r-1],i,\ms{bin\_vals})$
          
          \label{line:abv-broadcast}
          \SmallIf{$i=coord$}{}
          \label{line:coordbroad1}
          \WUntil{$\ms{bin\_vals}[r]=\{w\}$}          \EndWUntil
          \State $\lit{broadcast}(\COORD[r],w)$
          \EndSmallIf\label{line:coordbroad2}
          \WUntil{$\ms{bin\_vals}[r]\neq \emptyset \wedge \ms{timer}$ expired} \EndWUntil\label{line:wuntiphase1}
          
        }\EndPart
        \Part{$\blacktriangleright$ Phase 2}{
          \State $\ms{timer}\gets \ms{timeout}$ \Comment{reset timer}
          \SmallIf{$(\COORD[r],w)\in\ms{sig\_msgs} \wedge w\in\ms{bin\_vals}[r]$}{}%
          \label{line:delcoord}
          \State $\ms{aux}\gets \{w\}$\Comment{prioritize coordinator's value}
          \EndSmallIf
          \SmallElse{$\ms{aux}\gets \ms{bin\_vals}[r]$}\Comment{else use any received value}\label{line:binvalues1}
          \EndSmallElse
          \State $\lit{broadcast}(\ECHO[r],\ms{aux}[r])$\Comment{broadcast signed \ECHO message}\label{line:broadecho}
          \WUntil{($\ms{vals} = \lit{comp-vals}(\ms{sig\_msgs},\ms{bin\_vals},\ms{aux}))\neq \emptyset \wedge \ms{timer}$ expired}\EndWUntil\label{line:callcomputeval}
        }\EndPart
        \Part{$\blacktriangleright$ Decision phase}{
          \SmallIf{$|\ms{vals}|=1$}{$\ms{est}\gets\ms{vals}[0]$}
          \Comment{if only one, adopt as estimate}\label{line:dec1}
          \SmallIf{$\ms{est}=(r\,mod\,2)\,\wedge\,p_i$ not decided before}{}
          \State $\lit{decide}(\ms{est})$ \Comment{if parity matches, decide the estimate}\label{line:decide}
          \EndSmallIf
          \EndSmallIf
          \SmallElse{$\ms{est} \gets (r\,mod\,2)$}\Comment{otherwise, the estimate is the round's parity bit}\label{line:adoptestimate2}
          \EndSmallElse
          \State $\ms{cert}[r]\gets \lit{compute-cert}(\ms{vals},\ms{est},r,\ms{bin\_vals},$ $\ms{sig\_msgs})$
          \label{line:dec2}
          
          }\EndPart
        }\EndPart
        \Part{\textbf{Upon} receiving a signed message $\ms{s\_msg}$}{\label{line:dr1}
          \State $\ms{pofs}\gets \lit{check-conflicts}(\ms{\{s\_msg\}},\,\ms{sig\_msgs})$
          \Comment{returns $\emptyset$ or PoFs}\label{lin:smcc}
          \State $\lit{update-committee}(\ms{pofs})$\label{lin:smuc}\Comment{remove fraudsters}
          }\EndPart
          \Part{Upon receiving a certificate $\ms{cert\_msg}$}{
            \State $\ms{pofs}\gets \lit{check-conflicts}(\ms{cert\_msg},\,\ms{sig\_msgs})$          \Comment{returns $\emptyset$ or PoFs} 
            \State $\lit{update-committee}(\ms{pofs})$\Comment{remove fraudsters}
          }\EndPart
          \Part{Upon receiving a list of PoFs $\ms{pofs\_msg}$}{
            \SmallIf{$\lit{verify-pofs}(\ms{pofs\_msg})$}{}\Comment{if proofs are valid then}
            \State $\lit{update-committee}(\ms{pofs\_msg})$
            \Comment{remove fraudsters from committee}
            \EndSmallIf\label{line:dr2}
          }\EndPart
        \Part{Rules}{\label{lin:rul}
          \begin{enumerate}[leftmargin=* ,wide=\parindent]
          \item Every message that is not properly signed by the sender is discarded.
          \item Every message that is sent by $\lit{abv-broadcast}$ without a valid certificate after Round $1$, except for messages with value $1$ in Round $2$, are discarded.
          \item Every signed message received is stored in $\ms{sig\_msgs}$, including messages within certificates.\label{item:signed}
          \item Every time the timer reaches the timeout for a phase, and if that phase cannot be terminated, processes broadcast their current delivered signed messages for that phase (and all messages received for future phases and rounds) and reset the timer for that phase. These messages are added to the local set of messages and cross-checked for PoFs on arrival.\label{item:timer} 
          \end{enumerate}
         }\EndPart
         \algstore{alg:prot}
    \end{algorithmic}
  }
\end{algorithm}

The \protocol's AABC protocol is divided in two phases, after which a decision is taken. A key difference with Polygraph is that
when a timer for one of the two phases reaches its timeout, if a
process cannot terminate that phase yet, then it broadcasts its set of
signed messages for that phase and resets the timer, as detailed in Rule~\ref{item:timer}. This allows \protocol
to prevent deceitful processes from breaking termination by trying to
cause a disagreement and never succeeding. It is important that
processes wait for this timer before taking a decision for the phase,
since only waiting for that timer guarantees that all sent messages
will be received before the timer reaches its timeout, after GST.
Each process maintains an estimate (line~\ref{line:estimate}), initially given as input, and then proceeds in rounds executing the following phases:
\begin{enumerate}[leftmargin=* ,wide=\parindent]
  \item In the first phase, each process
broadcasts its estimate (given as input) via an accountable binary value reliable
broadcast (ABV-broadcast) (line~\ref{line:abv-broadcast}), which
we present in Algorithm~\ref{alg:helper}, lines~\ref{line:abvb-start}--\ref{line:abvb-end} and discuss in Section~\ref{sec:aabc}.
Decision and
$\lit{abv-broadcast}$ messages are discarded unless they come with a
certificate justifying them.

The protocol also uses a rotating coordinator (line~\ref{line:rotcoord}) per round which carries
a special \COORD message (lines~\ref{line:coordbroad1}-\ref{line:coordbroad2}). All processes wait until they deliver at
least one message from the call to $\lit{abv-broadcast}$ and until the
timer, initially set to $\Delta$, expires (line~\ref{line:wuntiphase1}). 
(Note that the bound on the message delays remains unknown due to the unknown GST.)
If a process delivers a message from
the coordinator (line~\ref{line:delcoord}), then it broadcasts an \ECHO message with the
coordinator's value and signature in the second phase (line~\ref{line:broadecho}). Otherwise, it
echoes all the values delivered in phase $1$ as part of the call
to $\lit{abv-broadcast}$ (line~\ref{line:binvalues1}).

\item In the second phase, processes wait till they receive $h(d_r)$ \ECHO
messages, as shown in the call to $\lit{comp-vals}$ (line~\ref{line:callcomputeval}), which returns
the set of values that contain these $h(d_r)$ signed \ECHO
messages. Function $\lit{comp-vals}$ is depicted in
Algorithm~\ref{alg:helper} (lines~\ref{line:comp-val-start}--\ref{line:comp-val-end}). Processes then
try to come to a decision in lines~\ref{line:dec1}-\ref{line:dec2}. As it was the case for phase $1$, when the
timer expires in phase $2$, all processes broadcast their current set
of \ECHO messages. Then, they update their committee if they detect deceitful processes through PoFs (lines~\ref{line:dr1}-\ref{line:dr2}) and recheck if
they reach the updated $h(d_r)$ threshold, after which they reset the
timer.

\item During the decision phase, if there is just one value returned by
$\lit{comp-vals}$ and that value's parity matches with the
round's parity, process $p_i$ decides it (line~\ref{line:decide}) and broadcasts the associated
certificate in the call to $\lit{compute-cert}$. If the parity
does not match then process $p_i$ simply adopts the value as the
estimate for the next round (line~\ref{line:dec1}). If instead there is more than one value
returned by $\lit{comp-vals}$ then $p_i$ adopts the round's
parity as next round's estimate
(line~\ref{line:adoptestimate2}). Adopting the parity as next round's
estimate helps with convergence in the next round, in this case where
processes are hesitating between two values.The call to
$\lit{compute-cert}$ (depicted at lines~\ref{lin:com-cer}--\ref{lin:com-cer-end} of Algorithm~\ref{alg:helper}) gathers the signatures justifying the
current estimate and broadcasts the certificate if the estimate was
decided in this round. 

\end{enumerate}
\begin{algorithm}[htbp]
  \caption{Helper Components.}
  \label{alg:helper}
  \smallsize{
    \begin{algorithmic}[1]
    \algrestore{alg:prot}
      \Part{$\lit{update-committee}(\ms{new\_pofs})$}{\Comment{function that removes fraudsters}\label{line:update-com-start}
        \SmallIf{$\ms{new\_pofs}\neq \emptyset\,\wedge\, \ms{new\_pofs}\not\subseteq \ms{local\_pofs}$}{}
        \State $\ms{new\_pofs}\gets \ms{new\_pofs}\backslash\ms{local\_pofs}$\Comment{consider only new PoFs}
        \State $\ms{local\_pofs}\gets \ms{local\_pofs}\cup\ms{new\_pofs}$\Comment{store new PoFs}
        \State $\lit{broadcast}(\POF,\,\ms{new\_pofs})$\Comment{broadcast new PoFs}
        \State $\ms{deceitful}\gets \ms{new\_pofs}.\lit{get\_processes()}$\Comment{get deceitful from PoFs}
        \State $\ms{new\_deceitful}\gets \ms{new\_deceitful}\backslash\ms{local\_deceitful}$
        \State $\ms{local\_deceitful}\gets \ms{local\_deceitful}\cup\ms{new\_deceitful}$      
        \State $N\gets N\backslash \{\ms{new\_deceitful}\}$; $n\gets |N|$   \Comment{remove new deceitful}
        \State $d_r\gets |\ms{local\_deceitful}|$\Comment{update number of detected deceitful}
        \State $h(d_r)\gets \lit{recalculate-threshold}(N,\,d_r)$        
        \State $\lit{recheck-certs-termination}()$
        \Comment{check termination of current phase}\label{line:recheccerts}
          \State $\lit{reset-current-timer()}$\Comment{reset timer of current phase}\label{line:resetcurt}
          \EndSmallIf \label{line:update-com-end}
        }\EndPart
        \Statex 
        \Part{$\lit{abv-broadcast}(\MSG,\ms{val},\ms{cert},i,\ms{bin\_vals})$}{ \label{line:abvb-start}
          \State $\lit{broadcast}(\BVALECHO,\langle \ms{val},\ms{cert},i\rangle)$\Comment{broadcast message}\label{line:bvechobroad1}
          \SmallIf{ $r=3$ \textbf{or} $(r=2$ \textbf{and} $\ms{val}=1)$}{}
          \State discard all messages received without a valid certificate
          \EndSmallIf
          \Upon{receipt of $(\BVALECHO,\langle v,\cdot,j\rangle )$}
          \SmallIf{$(\BVALECHO,\langle v,\cdot,\cdot\rangle)$ received from $\floor{\frac{n-q-t}{2}}-d_r+1$
            \\ distinct processes \textbf{and} $(\BVALECHO,\langle v,\cdot, i\rangle)$ not yet broadcast}{}
          \State Let $\ms{cert}$ be any valid certificate $cert$ received in these messages
          \State $\lit{broadcast}(\BVALECHO,\langle v,\ms{cert},i\rangle)$
          \Comment{see Lemma~\ref{lem:aabv-ter}}\label{line:bvechobroad2}
          \EndSmallIf
          \SmallIf{$(\BVALECHO,\langle v,\cdot,\cdot\rangle)$ received from $h(d_r)$ distinct processes \textbf{and}\\ $(\BVALREADY,\langle v,\cdot, \cdot\rangle)$ not yet broadcast}{}
          \State Let $\ms{cert}$ be any valid certificate $cert$ received in these messages

          \State Construct $\ms{bv\_cert}$ a certificate with $h(d_r)$ signed $\BVALECHO$

          \label{line:bvreadycons}
          \State $\ms{bin\_vals}\gets\ms{bin\_vals}.\lit{add}(\BVALREADY,\langle v,\ms{cert},j,\ms{bv\_cert}\rangle)$\label{line:bvdel1}

          \State $\lit{broadcast}(\BVALREADY,\langle v,\ms{cert},j,\ms{bv\_cert}\rangle)$
          \EndSmallIf
          \SmallIf{$(\BVALREADY,\langle v,cert,j,\ms{bv\_cert}\rangle)$ received from $1$ process}{}\label{line:bvreadyrec}
          \State $\ms{bin\_vals}\gets\ms{bin\_vals}.\lit{add}(\BVALREADY,\langle v,\ms{cert},j,\ms{bv\_cert}\rangle)$\label{line:bvdel2}

          \SmallIf{$(\BVALREADY,\langle v,cert,j,\ms{bv\_cert}\rangle)$ not yet broadcast}{}
          \State $\lit{broadcast}(\BVALREADY,\langle \ms{val},\ms{cert},i,\ms{bv\_cert}\rangle)$
          \EndSmallIf
          \EndSmallIf
          \EndUpon  \label{line:abvb-end}
        }\EndPart
        \Statex 
          \Part{$\lit{comp-vals}(\ms{msgs},\ms{b\_set},\ms{aux\_set})$}{ \label{line:comp-val-start}
            \Comment{check for termination of phase $2$}
            \State \textbf{If }$\exists S\subseteq \ms{msgs}$ where the following conditions hold:
            \State $\>\>$ $(i)\;|S|$ contains $h(d_r)$ distinct $\ECHO[r]$ messages        
            \State $\>\>$ $(ii)\;\ms{aux\_set}$ is equal to the set of values in $S$    \Comment{$h(d_r)$ with same est}
            \State $\>$\textbf{then return}$(\ms{aux\_set})$
            \State \textbf{Else If }$\exists S\subseteq \ms{msgs}$ where the following conditions hold:

            \State $\>\>$ $(i)\;|S|$ contains $h(d_r)$ distinct $\ECHO[r]$ messages
            \State $\>\>$ $(ii)\;$Every value in $S$ is in $\ms{b\_set}$    \Comment{$h(d_r)$ messages with different est}
            \State $\>$\textbf{then return}$(V=$ the set of values in $S)$
            \State \textbf{Else return}$(\emptyset)$ \Comment{else not ready to terminate} \label{line:comp-val-end}
            
          }\EndPart
          \Statex 
          \Part{$\lit{compute-cert}(\ms{vals},\ms{est},r,\ms{bin\_vals},\ms{msgs})$}{\label{lin:com-cer} \Comment{compute and send cert}

            \SmallIf{$\ms{est}=(r\,mod\,2)$}{} 
            \SmallIf{$r>1$}{}
            \State $\ms{to\_return}\gets(\ms{cert}$ : $(\EST[r],\langle v, \ms{cert},\cdot\rangle)\in\ms{bin\_vals})$
            \EndSmallIf
            \SmallElse{$\ms{to\_return}\gets(\emptyset)$}
            \EndSmallElse
            \EndSmallIf
            \SmallElse{ $\ms{to\_return}\gets(h(d_r)$ signed msgs containing only $\ms{est})$}
            \EndSmallElse
            \SmallIf{$\ms{vals}=\{(r\,mod\,2)\}\,\wedge\,$no previous decision by $p_i$}{}
            \State $cert[r]\gets h(d_r)$ signed messages containing only $r\,mod\,2$
            \State $\lit{broadcast}(\ms{est},r,i,\ms{cert}[r])$\Comment{broadcast decision}
              \EndSmallIf
              \State $\textbf{return}(\ms{to\_return})$ \label{lin:com-cer-end}
            }\EndPart 
    \algstore{alg:helper}
    \end{algorithmic}
  }
\end{algorithm}

\mypar{Detecting and removing deceitful processes} Upon receiving a signed message, non-faulty processes check if the
received message conflicts with some previously delivered message in
storage in $\ms{sig\_msgs}$ by calling $\lit{check-conflicts}$
(line~\ref{lin:smcc}). This function returns $\ms{pofs}=\emptyset$ if
there are no conflicting messages, or a list $\ms{pofs}$ of PoFs otherwise. Then, at line~\ref{lin:smuc},
non-faulty processes call $\lit{update-committee}$ (depicted at lines~\ref{line:update-com-start}--\ref{line:update-com-end} of
Algorithm~\ref{alg:helper}) to remove the $|\ms{pofs}|$ detected deceitful
processes at runtime. In the call to
$\lit{update-committee}$, process $p_i$ removes all processes that are
proven deceitful via new PoFs, and updates the committee $N$, its size
$n$, and the voting threshold $h(d_r)$. After that, $p_i$ rechecks all
delivered messages in that phase in case it can now terminate the
phase with the new threshold $h(d_r)$ (and after filtering out
messages delivered by the $d_r$ removed deceitful processes) by calling $\lit{recheck-certs-termination}()$ in line~\ref{line:recheccerts} of Algorithm~\ref{alg:helper}. Finally,
it resets the timer for the current phase by calling $\lit{reset-current-timer()}$ in line~\ref{line:resetcurt} of Algorithm~\ref{alg:helper}.

\mypar{Termination and agreement of Basilic's AABC} We
show the detailed proofs of agreement and termination in
Lemmas~\ref{lem:aabc-agr} and~\ref{lem:aabc-ter}. The idea is that 
removing deceitful processes has no effect on agreement, while it
facilitates termination, since the threshold $h(d_r)=h_0-d_r$
decreases the initial threshold $h_0$ with the number of removed
deceitful processes. Also, since all non-faulty
processes broadcast their delivered PoFs and thanks to the property of accountability, eventually all non-faulty
processes agree on the same set of removed deceitful processes.

Then, if a
process $p_i$ terminates broadcasting certificate $cert_i$ while
another process $p_j$ already removed newly detected deceitful
processes $new\_d_r$ present in $cert_i$, then
$|cert_i|-new\_d_r\geq h(d_r+new\_d_r)$ by construction. As such, either a non-faulty process terminates and then all
subsequent non-faulty processes can terminate, even after removing more deceitful
processes, or they all eventually reach a scenario where all deceitful
processes are detected $d_r=d$ and removed, after which they all terminate.



Note that removing processes at runtime can result in rounds whose
coordinator is already removed. For the sake of correctness, we do not
change the coordinator for that round even if it has already been
removed. This guarantees that all non-faulty processes eventually reach a
round in which they all agree on the same coordinator, which is a
non-faulty process. If this round is the first after GST and after all
deceitful processes have been removed from the committee, then non-faulty
processes will reach agreement.

\mypar{Accountable Binary Value Broadcast} The
ABV-broadcast that we present in Algorithm~\ref{alg:helper} is inspired from the E
protocol presented by Malkhi et al.~\cite{malkhi1997secure} and the binary value
broadcast presented in Polygraph~\cite{civit2019techrep,civit2021}. If non-faulty processes add a value $v$ to $\ms{bin\_vals}$ (lines~\ref{line:bvdel1} and~\ref{line:bvdel2}) as a result of the ABV-broadcast protocol, we say that they \textit{ABV-deliver} $v$. Processes exchange two types of messages during ABV-broadcast: \BVALECHO and \BVALREADY messages. \BVALECHO messages are
signed and must come with a valid certificate $\ms{cert}_i$ justifying
the value, as shown in lines~\ref{line:bvechobroad1} and~\ref{line:bvechobroad2}. \BVALREADY messages carry the same information as \BVALECHO
messages plus an additional certificate $\ms{bv\_cert}$ containing
$h(d_r)$ \BVALECHO messages justifying the \BVALREADY message, constructed in line~\ref{line:bvreadycons}. This way, as
soon as a process receives a \BVALREADY message with a value (line~\ref{line:bvreadyrec}), it already
obtains $h(d_r)$ \BVALECHO messages too, meaning it can ABV-deliver that
value adding it to $\ms{bin\_vals}$ (lines~\ref{line:bvdel1} and~\ref{line:bvdel2}). Non-faulty processes broadcast signed \BVALECHO messages for their estimate (line~\ref{line:bvechobroad1}) and for all values for which they receive at least $\floor{\frac{n-q-t}{2}}-d_r+1$ signed \BVALECHO messages from distinct processes. 

We prove in Lemma~\ref{lem:aabv-ter} that waiting for this many \BVALECHO messages for a value $v$ guarantees that all non-faulty processes ABV-deliver $v$. We defer the rest of the proofs to Appendix~\ref{sec:proofs}. In particular, we show that our
ABV-broadcast satisfies the following properties:
  (i) ABV-Termination, in that every non-faulty process eventually adds at least one value to $\ms{bin\_vals}$; (ii) ABV-Uniformity, in that non-faulty processes eventually add the same values to $\ms{bin\_vals}$; (iii) ABV-Obligation, in that if $\floor{\frac{n-q-t}{2}}-d_r+1$ non-faulty processes ABV-broadcast a value $v$, then all non-faulty processes ABV-deliver $v$; (iv) ABV-Justification, in that if a non-faulty process ABV-delivers a value $v$ then $v$ was ABV-broadcast by a non-faulty process; and (v) ABV-Accountability, in that every ABV-delivered value contains a valid certificate from the previous round.

We show in Lemma~\ref{lem:aabc-aac} that \protocol's AABC satisfies
AABC-\myproperty, but we defer the rest of the proofs of
\mypropertyadj binary consensus to the Appendix~\ref{sec:proofs}.

\begin{lemma}[AABC-\Myproperty]
  \label{lem:aabc-aac}
  \protocol's AABC satisfies \myproperty.
\end{lemma}
\begin{proof}
  We show that if a faulty process $p_i$ sends two conflicting messages to two
subsets $A,\,B\subseteq N$, each containing at least one non-faulty
process, then eventually all non-faulty processes terminate, or instead they
receive a PoF for $p_i$ and remove it from the committee, after which
they all terminate.

First, we observe that no process gets stuck in some round.
Process $p_i$ cannot get stuck in phase $1$ since,
by ABV-Termination (Lemma~\ref{lem:aabv-ter}), every non-faulty process eventually ABV-delivers a
value.

A process also does not get stuck waiting on phase $2$. First, notice
that every value that is included in an \ECHO message from a non-faulty
process is eventually delivered to $\ms{bin\_vals}$. Then, note that all
non-faulty processes eventually deliver $h(d_r)$ \ECHO messages, or
instead, when the timer expires, processes will exchange their \ECHO
messages and be able to construct PoFs and remove $d_r$ deceitful
processes that are preventing termination. In the latter case, after
removing all deceitful processes from the committee and updating the
threshold, they will deliver enough \ECHO messages to terminate
phase $2$, since $h(d_r)<n-q-t-d-d_r$.

Then, we show that all non-faulty processes always hold a valid certificate
to broadcast a proper message, which could otherwise prevent
termination of a phase during the ABV-broadcast in phase $1$. For an
estimate whose parity is the same as that of the finished round $r-1$,
process $p_i$ must have received a valid certificate for the round
(otherwise it would not have terminated such round). If the parity
matches, then it can always construct a valid certificate in round
$r-1$ from the delivered estimates.

As a result, all processes always progress infinitely in every round.
Consider the first round $r$ after GST where (i) the coordinator is
non-faulty and (ii) all deceitful processes have been detected and removed by all
non-faulty processes. In this case, every non-faulty
process will prioritize the coordinator’s value, adopting it
as their \ECHO message adding only that value. Hence, every
process adopts the same value, and hence decides either in round $r$ or
round $r+1$ (by Lemma~\ref{lem:aabc-aux}).

\end{proof}

\subsection{\Mypropertyadj Reliable Broadcast}
\label{sec:aarb}
Algorithm~\ref{alg:arb} shows \protocol's \textit{\mypropertyadj reliable broadcast} (AARB). The protocol is analogous to the secure broadcast presented in previous work~\cite{malkhi1997secure}, with the difference that we also introduce a timer that non-faulty processes use to periodically broadcast their set of delivered \ECHO messages, in order to detect deceitful processes. We refer of the process that starts the AARB protocol as the \emph{source}. The protocol starts when the source broadcast an \ECHO message with its proposed value $v$ (line~\ref{line:aarbecho1}). Upon delivering that message, all non-faulty processes also broadcast a signed \ECHO message with $v$ (line~\ref{line:aarbecho2}). Then, once a process $p_i$ delivers $h(d_r)$ distinct signed \ECHO messages for the same value $v$, $p_i$ first broadcasts a \READY message (line~\ref{line:aarb-broadcastReady1}) with a certificate containing the $h(d_r)$ \ECHO messages justifying $v$ (constructed in line~\ref{line:aarbconstructcert1}), and then AARB-delivers the value (line~\ref{line:aarb-deliver1}). The same occurs if instead a process delivers just one valid \READY message containing a valid certificate justifying it in lines~\ref{line:aarb-read1}-\ref{line:aarb-read2}.

As it occurs with \protocol's AABC protocol presented in Algorithms~\ref{alg:prot} and~\ref{alg:helper}, upon cross-checking newly received signed messages with previously delivered ones (lines~\ref{line:crosscheck1} and~\ref{line:crosscheck2}), non-faulty processes can detect deceitful faults and update the committee (lines~\ref{line:aarbupdcom1} and~\ref{line:aarbupdcom2}), removing them at runtime, by calling $\lit{update-committee}$. This can also occur when receiving a list of PoFs (line~\ref{line:aarbpofs}). Note that this is the same call to the same function as in the AABC protocol shown in Algorithm~\ref{alg:prot}, because non-faulty processes update the committee across the entire \protocol protocol, and not just for that particular instance of AARB or AABC where the deceitful process was detected.
We show in Appendix~\ref{sec:proofs} that \protocol's AARB protocol satisfies the following properties of \mypropertyadj reliable broadcast:
\begin{itemize}[leftmargin=* ,wide=\parindent]
  \item {\bf AARB-Unicity.} Non-faulty processes AARB-deliver at most one value.
  \item {\bf AARB-Validity.} Non-faulty processes AARB-deliver a value if it was previously AARB-broadcast by the source.
  \item {\bf AARB-Send.} If the source is non-faulty and AARB-broadcasts $v$, then non-faulty processes AARB-deliver $v$.
  \item {\bf AARB-Receive.} If a non-faulty process AARB-delivers $v$, then all non-faulty processes AARB-deliver $v$.
  \item {\bf AARB-Accountability.} If two non-faulty processes AARB-deliver distinct values, then all non-faulty processes receive PoFs of the deceitful behavior of at least $2h(d_r)-n$ processes including the source.
  \item{\bf AARB-\expandafter\MakeUppercase\myproperty.} Deceitful behavior does not prevent liveness.
  \end{itemize}

\begin{algorithm}[th]
  \caption{\protocol's AARB with initial threshold $h_0$.}
  \label{alg:arb}
  \smallsize{
    \begin{algorithmic}[1]
    \algrestore{alg:helper}
      \Part{\smallsize $\lit{AARB-broadcast_{h_0}(}v_i\lit{)}$}{\Comment{executed by the source}
        \State $\lit{broadcast(}\INIT, v_i)$ \Comment{broadcast to all}\label{line:aarbecho1}
        \Upon{\textbf{receiving $(\INIT, v_i)$ from $p_j$ \textbf{and}} \textbf{not having sent $\ECHO$:}}
        \State $\lit{boadcast(}\ECHO,v,j)$\Comment{echo value to all}\label{line:aarbecho2}
        \EndUpon
        \Upon{\textbf{receiving $h(d_r)$ $(\ECHO, v, j)$ \textbf{and}}\textbf{ not having sent a \READY:}}
        \State Construct $cert_i$ containing at least $h(d_r)$ signed msgs $(\ECHO, v, j)$
        \label{line:aarbconstructcert1}
        \State $\lit{broadcast(}\READY,v,cert_i,j)$\Comment{broadcast certificate}\label{line:aarb-broadcastReady1}
        \State $\lit{AARB-deliver(}v,j\lit{)}$ \Comment{AARB-deliver value}\label{line:aarb-deliver1}
        \EndUpon
        \Upon{\textbf{receiving} $(\READY,v,cert,j)$, \textbf{and} \textbf{not} \textbf{having sent a} $\READY$:}\label{line:aarb-read1}
        \SmallIf{$\lit{verify}(cert)=False$}{\textbf{continue}}\EndSmallIf
        \State Set $cert_i$ to be one of the valid certs received $(\READY,v,cert,j)$
        \State $\lit{broadcast(}\READY,v,cert_i,j)$\Comment{broadcast certificate}
        \State $\lit{AARB-deliver(}v,j\lit{)}$\Comment{AARB-deliver value}\label{line:aarb-read2}
        \EndUpon

        \Part{\textbf{Upon} receiving a signed message $\ms{s\_msg}$}{
          \State $\ms{pofs}\gets \lit{check-conflicts}(\ms{\{s\_msg\}},\,\ms{sig\_msgs})$\Comment{returns $\emptyset$ or PoFs} 

          \label{line:crosscheck1}
          \State $\lit{update-committee}(\ms{pofs})$\Comment{remove fraudsters}\label{line:aarbupdcom1}
          }\EndPart
          \Part{Upon receiving a certificate $\ms{cert\_msg}$}{
            \State $\ms{pofs}\gets \lit{check-conflicts}(\ms{cert\_msg},\,\ms{sig\_msgs})$\Comment{returns $\emptyset$ or PoFs} \label{line:crosscheck2}
            \State $\lit{update-committee}(\ms{pofs})$\Comment{remove fraudsters}\label{line:aarbupdcom2}
          }\EndPart
          \Part{Upon receiving a list of PoFs $\ms{pofs\_msg}$}{\label{line:aarbpofs}
            \SmallIf{$\lit{verify-pofs}(\ms{pofs\_msg})$}{}\Comment{if proofs are valid then}
            \State$\lit{update-committee}(\ms{pofs\_msg})$
            \Comment{exclude from committee}
            \EndSmallIf
          }\EndPart
      }\EndPart
      \Part{Rules}{
          \begin{enumerate}[leftmargin=* ,wide=\parindent]
          \item Processes broadcast their current delivered signed
$\INIT$ and $\ECHO$ messages once a timer $timer$, initially set to
$\Delta$, reaches $0$, and resets the timer to $\Delta$.
          \end{enumerate}
         }\EndPart
    \end{algorithmic}
  }
\end{algorithm}

\subsection{\protocol's fault tolerance in the BDB model}
We show in Figure~\ref{fig:fig31} the combinations of Byzantine, deceitful
and benign processes that \protocol tolerates, depending on the
initial threshold $h_0$. The solid lines represent the variation in tolerance to benign and
deceitful processes as the number of Byzantine processes varies for a
particular threshold. For example, for $h_0=\frac{2n}{3}$,
if $t=0$ then $d<\frac{n}{3}$ and $q<\frac{n}{3}$. As $t$
increases, for example to $t=\ceil{\frac{n}{6}}-1$, then
$d<\frac{n}{6}$ and $q<\frac{n}{6}$.
\begin{figure}[htp] \center
  \includegraphics[width=.7\textwidth]{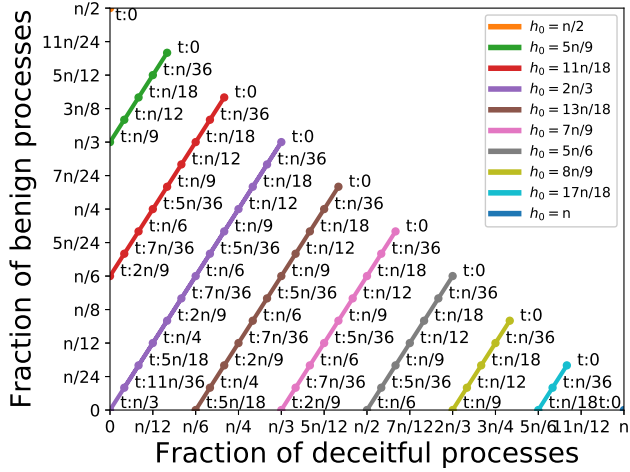}
  \caption{Combinations of benign, deceitful and Byzantine processes that \protocol tolerates, for an initial threshold $h_0$.}
  \label{fig:fig31}
\end{figure}

We compare our \protocol's fault tolerance with that of previous works
in Figure~\ref{fig:fig32}. In particular, we represent multiple values
of the initial threshold
$h_0\in\{5n/9,\,2n/3,\,3n/4,\,5n/6\}$ for Basilic. First, we show that classical
Byzantine fault-tolerant (BFT) protocols tolerate only the case $t<n/3$ with a blue triangle dot (\markerthree) in the
figure. This is the case of most partially synchronous BFT consensus protocols~\cite{crain2018dbft,civit2021, KADC07, YMR19, ranchal2020blockchain}. Notice that Zero-loss Blockchain~\cite{ranchal2020blockchain} (ZLB) also tolerates instead $d<5n/9$ and $3q+d<n$ faults, where $d$ and $q$ is the number of deceitful and benign faults, but that ZLB does not solve consensus for these bounds, and instead it recovers from disagreements. Second, we
represent Flexible BFT~\cite{MNR19} in their greatest fault tolerance
setting in partial synchrony. As we can see, such setting overlaps
with \protocol's initial threshold of $h_0=2n/3$. However, the difference lies in that
while \protocol tolerates all the cases in the solid line $h_0=2n/3$,
Flexible BFT only tolerates a particular dot of the line, set at the discretion of
each client. That is, Flexible BFT's clients must decide, for example,
whether they tolerate either $\ceil{2n/3}-1$ total faults, being none
of them Byzantine, or instead tolerate $\ceil{n/3}-1$ Byzantine faults, not tolerating any
additional fault. \protocol can however tolerate any range satisfying
both $h_0>\frac{n+d+t}{2}$ for safety and $h_0\leq n-q-t$ for
liveness, which allows our clients and servers to tolerate significantly more combinations of faults for one particular threshold $h_0\in(n/2,n]$. For this reason, we
represent the line of Flexible BFT as a dashed line, whereas
\protocol's lines are solid. For each initial voting threshold $h_0$, the maximum
number of Byzantine processes Basilic tolerates is $t<\min(2h_0-n,1-h_0)$,
which is obtained by setting $q=d=0$ and resolving both bounds for
safety and liveness.
  

\begin{figure}[htp]\center
  \includegraphics[width=.7\textwidth]{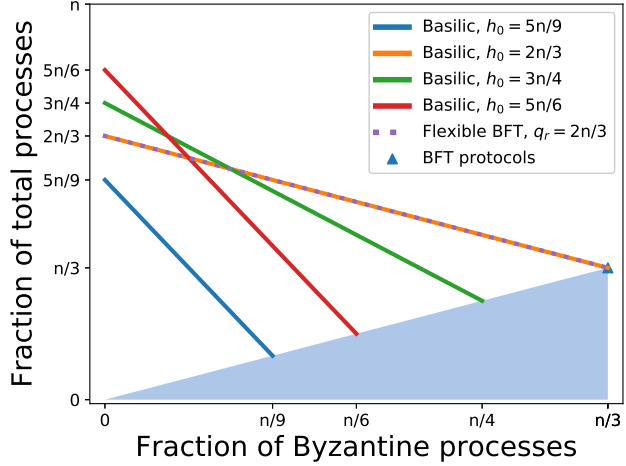}
  \caption{ Fraction of total processes, compared with fraction of Byzantine processes, for a particular initial threshold $h_0$ of the general \protocol protocol, compared with other works.}
  \label{fig:fig32}
\end{figure}
\subsection{\protocol's correctness}
We show in Lemma~\ref{lem:aac} that \protocol satisfies
\myproperty. We defer to Appendix~\ref{sec:proofs} the rest of the
proofs that show that the \protocol class of protocols solves the \problem
problem for the resilient optimal bounds of the impossibility results
shown in Section~\ref{sec:imp}. 
\begin{lemma}[\Myproperty]
  \label{lem:aac}
  \protocol satisfies \myproperty.
\end{lemma}
\begin{proof}
  We show that if a faulty process $p_i$ sends two conflicting
messages to two subsets $A,\,B\subseteq N$,
each containing at least one non-faulty process, then eventually all
non-faulty processes terminate, or instead they receive a PoF for $p_i$
and remove it from the committee, after which they all terminate.

  First, analogously to Lemma~\ref{lem:aabc-aac}, w.l.o.g. we treat
only the case $d_r=0$, since all conflicting messages that can be sent
in \protocol are messages of \protocol's AARB or AABC, that already
satisfy \myproperty (see Lemmas~\ref{lem:aabc-aac}
and~\ref{lem:aarb-aac}). This means that if $d_r>0$, then non-faulty
processes eventually update the committee and threshold, after which they recheck
if they hold enough signed messages to terminate. Next, we prove termination. By the
AARB-Send property (Lemma~\ref{lem:aarb-sen}), all non-faulty processes will eventually deliver
the proposals from non-faulty processes. Eventually all non-faulty
processes propose $1$ in all binary consensus whose index corresponds
to a non-faulty proposer, and by AABC-Validity decide $1$. Since
eventually $h(d_r)\leq n-q-d-t$ if enough $d_r$ prevent termination
and are thus detected and removed, we can conclude that at least
$h(d_r)$ binary consensus instances will terminate deciding $1$.

Once non-faulty processes decide $1$ on at least $h(d_r)$ proposals, they propose
$0$ to the rest, and by AABC-Termination (Lemma~\ref{lem:aabc-ter}) all remaining
binary consensus instances will terminate. Next, we show that for
every binary consensus upon which we decided $1$, at least one non-faulty
process AARB-delivered its associated proposal. For the sake of
contradiction, if no non-faulty process had AARB-delivered its associated
proposal, then all non-faulty processes would have proposed $0$, meaning by
AABC-Validity that the final decision of the binary consensus would have
been $0$, not $1$. As a result, by the AARB-Receive property (Lemma~\ref{lem:aarb-rec}),
eventually all non-faulty processes will deliver the proposal for all
binary consensus that they decided $1$ upon. Finally, processes
decide the value proposed by the proposer with the lower index.
\end{proof}

We summarize all proofs in the result shown in
Theorem~\ref{thm:consf1} to show that \protocol protocol with initial
threshold $h_0$ solves consensus if $d+t<2h_0-n$ and $q+t\leq
n-h_0$. This result translates in the \protocol class of protocols
solving consensus if $n>3t+d+2q$, as we show in
Corollary~\ref{cor:consf}.

\begin{theorem}[Consensus per threshold]
  \label{thm:consf1}
  The \protocol protocol with initial threshold $h_0$ solves the
\problem problem if $d+t<2h_0-n$ and $q+t\leq n-h_0$.
\end{theorem}

\begin{corollary}[Consensus]
  \label{cor:consf1}
  The \protocol class of protocols solves \problem if $n>3t+d+2q$.
\end{corollary}
\section{\protocol's complexity}
\label{sec:comps}
In this section, we show the time, message and bit
complexities of \protocol. We execute one instance of
\protocol's AARB reliable broadcast and of \protocol's AABC binary
consensus per process. We prove these complexities in the
appendix~\ref{sec:extcomp}. We summarize the complexities of the three protocols after
GST in Table~\ref{tab:gstcom}. We refer to the appendix~\ref{sec:extcomp} for an
analysis of complexities before GST.
\begin{table}[htp]
  \centering
  \setlength{\tabcolsep}{12pt}
\begin{tabular}{lrrr} 
\toprule 
Complexity & 
AARB & AABC & Basilic \\\midrule
 Time & $\mathcal{O}(1)$ & $\mathcal{O}(n)$ & $\mathcal{O}(n)$ \\
 Message & $\mathcal{O}(n^2)$ & $\mathcal{O}(n^3)$ & $\mathcal{O}(n^4)$ \\
 Bit& $\mathcal{O}(\lambda n^3)$& $\mathcal{O}(\lambda n^4)$& $\mathcal{O}(\lambda n^5)$\\
 \bottomrule
\end{tabular}

\caption{Time, message and bit complexities of \protocol
AARB, AABC and the general \protocol protocol, after GST.}
\label{tab:gstcom}
\end{table}

The complexities of \protocol after GST share the same asymptotic
complexity of other recent works that are not
\mypropertyadj~\cite{civit2019techrep, civit2021, forensics}, some of them not
being accountable either~\cite{CL02}, as we show in
Table~\ref{tab:bigtab}. This is because the adversary cannot prevent
termination of any phase. Thus, after GST, all processes can continue
to the next phase or terminate the protocol by the time the timer for
that phase expires, resulting in an execution equivalent to that of
Polygraph (apart from one additional message broadcast in
ABV-broadcast). In this table, naive \protocol represents the protocol
we show in Algorithm~\ref{alg:gen}, whereas the following row,
multi-valued \protocol, shows the analogous optimizations shown in
Polygraph and applicable to the \protocol protocol as
well~\cite{civit2021}. The row titled 'superblock' is the result of
applying the additional superblock
optimization~\cite{crain2018dbft,CNG18}. This optimization is only
available to protocols without a leader in which all processes propose
a value~\cite{crain2018dbft,civit2021}, and consists of deciding the
union of all values instead of the minimum of them. After these
optimizations, the resulting normalized bit complexity (i.e. per
decision) of \protocol is as low as those of other works that are only
accountable and not actively accountable, such as BFT
Forensics~\cite{forensics} or Polygraph~\cite{civit2021}. Furthermore,
since this is the lowest complexity to obtain
accountability~\cite{civit2021}, this means that this is also optimal
in the bit complexity. Note that other optimizations present in other
works, such as the possibility to obtain an amortized complexity of
$\mathcal{O(\lambda n^2)}$ in BFT Forensics per decision after $n$
iterations of the protocol~\cite{TG19}, is also possible in
\protocol's consensus protocol. Finally, an advantage of \protocol, as
well as of other leaderless protocols, compared to leader-based
works~\cite{forensics,TG19}, is that the distribution of proposals
scatters the bits throughout multiple channels of the network,
instead of bloating channels that have the leader as sender or
recipient, as previously noted~\cite{CNG18}.

Finally, not only are the rest of the protocols
in Table~\ref{tab:bigtab} not \mypropertyadj, but also this means that
they only solve consensus tolerating at most $t<n/3$ faults in the BDB
model, whereas \protocol with initial threshold $h_0=2n/3$ solves
consensus where $d+t<n/3$ and $q+t\leq n/3$ faults, hence tolerating
the strongest adversary among these proposals.

\begin{table}[htp]
  \centering
  \caption{Complexities of \protocol compared to other works.}
  \label{tab:bigtab}
  \setlength{\tabcolsep}{5pt}
  \hspace{-1em}
\begin{tabular}{lllll}
\toprule
Algorithm & Msgs & Bits & Acc. & Actacc. \\
\midrule
PBFT~\cite{CL02} & $\mathcal{O}(n^3)$ & $\mathcal{O}(\lambda n^4)$ & \no & \no \\
Tendermint~\cite{Buc16} & $\mathcal{O}(n^3)$ & $\mathcal{O}(\lambda n^3)$ & \no & \no \\
HotStuff~\cite{TG19} & $\mathcal{O}(n^2)$ & $\mathcal{O}(\lambda n^2)$ & \no & \no \\
DBFT superblock~\cite{crain2018dbft} & $\mathcal{O}(n^3)$ & $\mathcal{O}(n^3)$ & \no & \no \\
  
\midrule
  BFT Forensics~\cite{forensics} & $\mathcal{O}(n^2)$ & $\mathcal{O}(\lambda n^3)$ & \yes & \no \\
Polygraph's binary~\cite{civit2021} & $\mathcal{O}(n^3)$ & $\mathcal{O}(\lambda n^4)$ & \yes & \no \\  
Naive Polygraph~\cite{civit2021} & $\mathcal{O}(n^4)$ & $\mathcal{O}(\lambda n^5)$ & \yes & \no \\
Polygraph Multi-v.~\cite{civit2021} & $\mathcal{O}(n^4)$ & $\mathcal{O}(\lambda n^4)$ & \yes & \no \\
  
Polygraph superblock.~\cite{civit2021} & $\mathcal{O}(n^3)$ & $\mathcal{O}(\lambda n^3)$ & \yes & \no \\
\midrule
\protocol's AABC & $\mathcal{O}(n^3)$ & $\mathcal{O}(\lambda n^4)$ & \yes & \yes \\
Naive \protocol & $\mathcal{O}(n^4)$ & $\mathcal{O}(\lambda n^5)$ & \yes & \yes \\
Multi-valued \protocol & $\mathcal{O}(n^4)$ & $\mathcal{O}(\lambda n^4)$ & \yes & \yes \\
\protocol superblock & $\mathcal{O}(n^3)$ & $\mathcal{O}(\lambda n^3)$ & \yes & \yes \\
\bottomrule
\end{tabular}
\end{table}
\section{Solving Eventual Consensus with \protocol}
\label{sec:ec}
In this section, we adapt \protocol to solve eventual consensus in the BDB model, and then prove that the \protocol protocol is resilient optimal.
The eventual consensus ($\Diamond$-consensus) abstraction~\cite{Dubois2015} captures eventual agreement among all participants. It exports, to every process $p_i$, operations
$\lit{proposeEC}_1, \lit{proposeEC}_2, ...$ that take multi-valued arguments (non-faulty processes propose valid values) and return multi-valued responses. Assuming that, for all $j \in N$,
every process invokes $\lit{proposeEC}_j$ as soon as it returns a response to $\lit{proposeEC}_{j-1}$, the abstraction guarantees that, in every admissible run, there exists $k \in N$, such that the following properties are satisfied:
\begin{itemize}
\item {\bf $\Diamond$-Termination.} Every non-faulty process eventually returns a response to $\lit{proposeEC}_j$
for all $j\in N$.
\item {\bf $\Diamond$-Integrity.} No process responds twice to $\lit{proposeEC}_j$ for all $j\in N$.
\item {\bf $\Diamond$-Validity.} Every value returned to $\lit{proposeEC}_j$ was previously proposed to $\lit{proposeEC}_j$ for all $j\in N$.
\item {\bf $\Diamond$-Agreement.} No two non-faulty processes return different values to $\lit{proposeEC}_j$
for all $j\geq k$.
\end{itemize}

We detail now \emph{$\Diamond$-\protocol (BEC)}, an adaptation of \protocol for the $\Diamond$-consensus problem. Process $p_i$ executes $\Diamond$-\protocol with the following steps:
\begin{enumerate}[leftmargin=* ,wide=\parindent]
\item BEC first executes $\lit{\protocol-gen-propose_{h_0}(}v_i\lit{)}$, whose output is returned by $p_i$ as BEC's output of $\lit{proposeEC}_0$.
\item If $p_i$ finds no disagreement between operations $k$ and $k'$, then for all operations $\lit{proposeEC}_j,\,k'>j\geq k$, the output is that of $\lit{proposeEC}_{j-1}$.
\item If $p_i$ finds a new disagreement at operation $j$ for some index $r\in[0,n-1]$, then:
  \begin{enumerate}[leftmargin=* ,wide=\parindent]
    \item If the disagreement is between AARB-delivered values, BEC resolves it as follows: let $(\EST, \langle u,r\rangle)$ be the value that differs with the locally AARB-delivered value $(\EST, \langle v,r\rangle)$, then, for $\lit{proposeEC}_j$, $p_i$ applies $y=\lit{min(}v,u\lit{)}$ to the disagreeing value. Next, if the output of $\lit{proposeEC}_{j-1}$ was $v$, $p_i$ replaces the AARB-delivered value with $y$, and outputs $y$ instead for $\lit{proposeEC}_{j}$. 
    \item If the disagreement is between values 1 and 0 decided at AABC's protocol, then $p_i$ sets $\ms{bin-decisions}[r]$ to 1. Then, $p_i$ recalculates if the minimum decided value changed after adding this binary decision (i.e., re-execute lines~\ref{line:genuni1}-\ref{line:genuni} of Algorithm~\ref{alg:gen}), and output this decision for $\lit{proposeEC}_{j}$.
    \item Finally, $p_i$ broadcasts the values (and certificates) of all the disagreements that $p_i$ has not yet broadcast.
    \end{enumerate}

\end{enumerate}
We show in Theorem~\ref{thm:ec} that $\Diamond$-\protocol with initial threshold $h_0$ solves the $\Diamond$-consensus
problem if $d+t<h_0$ and $q+t<n-h_0$, where $t,\,d$ and $q$ are the numbers
of Byzantine, deceitful and benign processes, respectively, and $h_0$ the initial threshold. This means that the $\Diamond$-\protocol class of protocols solves $\Diamond$-consensus for any combination of $t,\,d$ and $q$ Byzantine, deceitful and benign processes, respectively,
such that $2t+d+q<n$, as we show in Corollary~\ref{cor:dec1}.
\begin{theorem}[$\Diamond$-Consensus per threshold]
  \label{thm:ec}
  The $\Diamond$-\protocol protocol with initial threshold $h_0$ solves the $\Diamond$-consensus
problem if $d+t<h_0$ and $q+t<n-h_0$.
\end{theorem}
\begin{proof}
  $\Diamond$-Integrity is trivial. The bound $q+t<n-h_0$ is proven in
Corollary~\ref{cor:impter}: $\Diamond$-\protocol starts by executing
\protocol, which does not terminate unless $q+t<n-h_0$, satisfying
$\Diamond$-Termination. $\Diamond$-Validity derives immediately from \protocol's proof of validity (Lemma~\ref{lem:val}).

We only have left to prove $\Diamond$-Agreement. If $d+t<h_0$ then all
valid certificates contain at least one non-faulty process. This means
that the number of disagreements is finite. Then, since non-faulty
processes broadcast all disagreements they find (and their
corresponding valid certificates), all non-faulty processes will
eventually find all disagreements. Also, all non-faulty processes will find all disagreements of \protocol by its accountability property (Lemma~\ref{lem:aac}). Let us consider that all non-faulty
processes, except $p_i$, have already found and treated all
disagreements (as specified by the $\Diamond$-\protocol protocol). Suppose
that $p_i$ finds the last disagreement at the start of operation
$\lit{proposeEC}_{k-1}$ for some $k>0$. Then, for all $j\geq k$, no
two non-faulty processes return different values to
$\lit{proposeEC}_k$, satisfying $\Diamond$-Agreement.
\end{proof}
\begin{corollary}[$\Diamond$-Consensus]
  \label{cor:dec1}
  The \protocol class of protocols solves $\Diamond$-consensus if $n>2t+d+q$.
\end{corollary}
\section{Related Work}
\label{sec:relw}
  Accountability has been proposed for distributed systems in
PeerReview~\cite{HKD07} and particularly for the problem of consensus
in Polygraph~\cite{CGG20}. ZLB~\cite{ranchal2020blockchain} extends Polygraph to tolerate
up to $5n/9$ deceitful faults for $\Diamond$-consensus, but tolerates only $t<n/3$ for consensus. This work leverages accountability to
replace deceitful processes by new processes. Unfortunately, they
require deceitful processes to eventually
stop trying to cause a disagreement. Flexible BFT~\cite{MNR19} offers a failure model and theoretical results to
tolerate $\lceil 2n/3\rceil -1$ alive-but-corrupt (abc) processes. An
abc process behaves maliciously only if it knows it can violate
safety, and behaves correctly otherwise. This is an even stronger
assumption than ZLB's deceitful faults eventually behaving
correctly. Additionally, their fault tolerance requires a commitment
from clients to not tolerate a single Byzantine fault in order to
tolerate $\lceil 2n/3\rceil -1$ abc faults, or to instead tolerate no
abc faults if clients decide to tolerate $t=\lceil n/3\rceil -1$ Byzantine faults.
Neu et al.'s ebb-and-flow system~\cite{ebbnflow} is available in
partial synchrony for $t<n/3$ and satisfies finality in synchrony for
$t<n/2$. They also motivate the need for the BDB model in their recent
accountability-availability
dilemma~\cite{neu2021availability}. Sheng et al.\cite{forensics}
characterize the forensic support of a variety of Blockchains. Unfortunately,
none of these works tolerate $q=\ceil{\frac{n}{3}}-1$ benign and even
$d=1$ deceitful faults, or $d=\ceil{\frac{n}{3}}-1$ and even $q=1$
benign fault, a direct consequence of them not satisfying \myproperty.

Upright~\cite{CKL09} tolerates $n=2u+r+1$ faults, where $u$ and $r$
are the numbers of commission and omission faults,
respectively. Upright tolerates $n/3$ commission faults or instead
$n/2$ omission faults, falling short of \protocol's $q+d<2n/3$
deceitful and benign faults or $t<n/3$ Byzantine faults
tolerated. Upright does also not tolerate more faults for commission
than the lower bound for BFT consensus. Anceaume et
al.~\cite{anceaume2020finality} tolerate $t<n/2$ Byzantine faults for
the problem of eventual consensus, at the cost of not tolerating even
$t=1$ Byzantine fault for deterministic consensus. Our \protocol class also tolerates this case if $h_0$ is set to $h_0=\floor{\frac{n}{2}}+1$, but \protocol also tolerates more cases by a just changing the initial threshold $h_0$. 

Although \protocol is, to the best of our knowledge, the first
protocol tolerating $n>3t+d+2q$ in the BDB model, and
despite this fault tolerance deriving from the property of
\myproperty, previous works already try to discourage misbehavior by
threatening with slashing a deposit or removing a faulty process from
the committee, or both. Ranchal-Pedrosa et al. propose the Huntsman
protocol~\cite{ranchal2021agreement}, an accountable consensus protocol tolerating
up to $k$ rational players and $t$ Byzantine players causing a
disagreement by threatening deviant rationals and rewarding those who expose the deviants, for $n>\max(\frac{3}{2}k+3t,2(k+t))$. Freitas de
Souza et al.~\cite{SKRP21} provide an asynchronous implementation of an accountable lattice agreement protocol. Shamis et al.~\cite{shamis2021pac} store signed
messages in a dedicated ledger so as to punish processes in case of
misbheavior. The Casper~\cite{Casper} algorithm
incurs a penalty in case of double votes but does not ensure
termination when $t<n/3$. Although Tendermint~\cite{BKM18} aims at slashing processes, it is not accountable. SUNDR~\cite{LKM04} requires cross-communication between non-faulty clients
to detect failures. FairLedger~\cite{levari2019fairledger} requires synchrony to detect faulty processes.

\section{Conclusion}
\label{sec:con}
In this paper, we have shown that it is impossible to solve consensus
in the BDB model against an adversary controlling
$n>3t+d+2q$, where $t,\,d,$ and $q$ are the number of Byzantine, deceitful and benign processes, respectively. We then present our \protocol
class of protocols, the first class of resilient optimal protocols for
the consensus problem in the BDB model. \protocol solves
\mypropertyadj consensus tolerating any combination of $t,\, d$ and $q$
Byzantine, deceitful and benign processes, respectively, satisfying
$h_0>\frac{n+d+t}{2}$ for safety and $h_0\leq n-q-t$ for liveness,
given an initial voting threshold $h_0$. We prove this result to be
resilient optimal per voting threshold. Additionally, for the same
voting threshold, \protocol also solves eventual consensus if
$h_0>d+t$ and $h_0\leq n-q-t$. We show that \protocol's complexities
are comparable to state-of-the-art accountable consensus protocols
that tolerate less faults.
 
\bibliography{consensus-many-faults} 

\newpage 
\appendix
\subsection{\protocol Proofs}
\label{sec:proofs}
In this section, 
we prove the rest of the properties of \protocol, including its ABV-broadcast, AABC and AARB protocols.
\subsubsection{Accountable binary value broadcast}
We first start with the properties that ABV-broadcast satisfies. We say process $p_i$ ABV-broadcasts value $v$ to refer to $p_i$ sending a \BVALECHO message containing $v$ and a valid certificate justifying $v$.
We prove ABV-termination in Lemma~\ref{lem:aabv-ter}, ABV-uniformity in Lemma~\ref{lem:aabv-uni}, ABV-obligation in Lemma~\ref{lem:aabv-obl}, ABV-justification in Lemma~\ref{lem:aabv-jus}, and ABV-accountability in Lemma~\ref{lem:aabv-acc}.

\begin{lemma}[ABV-Termination]
  \label{lem:aabv-ter}
   Every non-faulty process eventually adds at least one value to $\ms{bin\_vals}$.
 \end{lemma}
 \begin{proof}
   Note that all non-faulty processes broadcast a \BVALECHO message with
value $v$ when they receive $\floor{\frac{n-q}{2}}-d_r+1$ \BVALECHO messages
with $v$. First, let us consider that $t=d=0$, in that case,
non-faulty processes broadcast a \BVALECHO message with $v$ if they receive
$\floor{\frac{n-q-t}{2}}+1$ \BVALECHO messages with $v$. Also recall that $v\in\{0,1\}$. As such,
let us consider a partition of non-faulty processes $A,\,B\subseteq N$
such that $A\cap B=\emptyset$, and let us consider that processes in
$A$ initially sent a \BVALECHO message with $v=0$ while processes in
$B$ sent a \BVALECHO message with $v=1$. It is clear that $|A|+|B|\geq
n-q-t$ and thus either $|A|\geq \floor{\frac{n-q-t}{2}}+1$ or $|B|\geq
\floor{\frac{n-q-t}{2}}+1$. W.l.o.g. let us assume that $|A|\geq
\floor{\frac{n-q-t}{2}}+1$, then processes in $|B|$ eventually receive enough
\BVALECHO messages with value $v=0$ to also broadcast a \BVALECHO
message with $v=0$. Thus, since $h(d_r)\leq n-q-t-d_r$, eventually
all non-faulty processes receive enough \BVALECHO messages to add at
least the value $0$ to $\ms{bin\_vals}$.

Suppose instead that $d>0$ and $t=0$. Then, if the $d_r\leq d$ deceitful
processes that behave deceitful at a particular phase are enough to prevent termination, this
means that $d_r$ processes have sent at least two conflicting messages
to at least two non-faulty processes. As such, when the timer expires
and non-faulty processes broadcast their received signed \BVALECHO
messages, all non-faulty processes will eventually receive enough \BVALECHO messages to send a \BVALECHO message (analogously to case $d=0$). Thus, the
case $d>0$ is analogous to the case $d=0$ since \BVALECHO messages are relayed when timer expires, and
we have proven in the previous paragraph that termination is
guaranteed in that case. The same analogy takes place if $t>0$.

Note additionally that if $d_r$ detected deceitful processes have been removed, then the thresholds
decrease by the same factor $d_r$, preserving termination.

\end{proof}

\begin{lemma}[ABV-Uniformity]
  \label{lem:aabv-uni}
  If a non-faulty process $p_i$ adds value $v$ to the set $\ms{bin\_vals}$, then all other non-faulty processes also eventually add $v$ to their local set $\ms{bin\_vals}$.
\end{lemma}
\begin{proof}
  This proof is straightforward: $p_i$ adds $v$ to the set $\ms{bin\_vals}$ if it holds $h(d_r)$ signed \BVALECHO messages with $v$. In that case, it also constructs a certificate $\ms{bv\_cert}$ with these messages and broadcasts $\ms{bv\_cert}$ as part of the \BVALREADY with $v$ before adding $v$ to $\ms{bin\_vals}$. Therefore, all other non-faulty processes will eventually receive $p_i$'s \BVALREADY message along with $\ms{bv\_cert}$ containing enough \BVALECHO messages to also add $v$ to their local $\ms{bin\_vals}$. Finally, recall that all non-faulty processes broadcast their \BVALREADY message before adding $v$ to $\ms{bin\_vals}$, which solves the case that $p_i$ is faulty and sends \BVALREADY only to a subset of the non-faulty processes.
  \end{proof}
\begin{lemma}[ABV-Obligation]
  \label{lem:aabv-obl}
  If $\floor{\frac{n-q-t}{2}}-d_r+1$ non-faulty processes ABV-broadcast a value $v$, then all non-faulty processes ABV-deliver $v$.
\end{lemma}
\begin{proof}
This proof is analogous to that of Lemma~\ref{lem:aabv-ter}.
\end{proof}
\begin{lemma}[ABV-Justification]
  \label{lem:aabv-jus}
  If process $p_i$ is non-faulty and ABV-delivers $v$, then $v$ has been ABV-broadcast by some non-faulty process.
\end{lemma}
\begin{proof}
  Assume first $t=0$ and suppose the contrary: $p_i$ ABV-delivers $v$ and all non-faulty
processes ABV-broadcast $v',\,v\neq v'$. Since benign processes may
either send $v'$ to a subset of the non-faulty processes or nothing at
all, this means that $d-d_r > \floor{\frac{n-q }{2}}-d_r+1$ for
deceitful alone to be able to make $p_i$ ABV-deliver $v$. But using
the bound $d-d_r<n-h(d_r)$ we obtain that $q\geq 2h(d_r)-n$, which
contradicts our assumption on the number of benign faults (i.e. the
bound $q<2h(d_r)-n$). As a
result, it follows that at least some non-faulty process must have ABV-broadcast $v$.
The prove is analogous if $t>0$.
\end{proof}
 \begin{lemma}[ABV-Accountability]
   \label{lem:aabv-acc}
   If process $p_i$ adds value $v$ to $\ms{bin\_vals}$ then associated with $v$ is a valid certificate $\ms{cert}$ from the previous round.
 \end{lemma}
 \begin{proof}
   Since every \BVALECHO and \BVALREADY message without a valid
certificate is discarded, it follows immediately that when a value $v$
is added to $\ms{bin\_vals}$ then $p_i$ has access to a valid
certificate.
   \end{proof}
   \subsubsection{\Mypropertyadj reliable broadcast}
   In this section, we prove the properties of \protocol's reliable broadcast, AARB. In
these proofs, we refer to $p_s$ as the source of the
AARB-broadcast, i.e. the process that sends the \INIT message. We prove AARB-unicity in Lemma~\ref{lem:aarb-uni}, AARB-validity in Lemma~\ref{lem:aarb-val}, AARB-send in Lemma~\ref{lem:aarb-sen}, AARB-Receive in Lemma~\ref{lem:aarb-rec}, AARB-accountability in Lemma~\ref{lem:aarb-acc} and AARB-\myproperty in Lemma~\ref{lem:aarb-aac}.
\begin{lemma}[AARB-Unicity]
  \label{lem:aarb-uni}
  Non-faulty processes AARB-deliver at most one value. 
\end{lemma}
\begin{proof}
  By construction all non-faulty processes AARB-deliver at most one value.
\end{proof}
\begin{lemma}[AARB-Validity]
  \label{lem:aarb-val}
  If non-faulty process $p_i$ AARB-delivers $v$, then $v$ was AARB-broadcast by $p_s$. 
\end{lemma}
\begin{proof}
  Process $p_i$ AARB-delivers $v$ if it receives $h(d_r)$ messages $\langle
\ECHO, v,\cdot,\cdot\rangle$. Non-faulty processes only send an \ECHO message
for $v$ if they receive $\langle \INIT,\,v\rangle$. Thus, since $d+t<h(d_r)$,
$p_s$ AARB-broadcast $v$ to at least one non-faulty process.
\end{proof}
\begin{lemma}[AARB-Send]
  \label{lem:aarb-sen}
  If $p_s$ is non-faulty and AARB-broadcasts $v$, then all non-faulty
processes eventually AARB-deliver $v$.
\end{lemma}
\begin{proof}
  Deceitful processes either broadcast $v$
or multicast $v'$ to a partition $A$ and $v$ to a partition $B$. In
the first case (in which all deceitful behave like non-faulty processes),
since the number of benign and Byzantine processes is $q+t<n-h(d_r)$ it follows that
at least $h(d_r)$ non-faulty processes will echo $v$, being that
enough for all processes to eventually AARB-deliver it.

Consider instead some $d_r\leq d+t$ deceitful processes behave deceitful
echoing different messages to two different partitions each containing
at least one non-faulty process. Then when the timer expires and
non-faulty processes exchange their delivered \ECHO messages, all
processes will update their committee removing the $d_r$ detected
deceitful. Thus, since processes also recalculate the thresholds and
recheck them after updating the committee, this case becomes the
aforementioned case where no deceitful process behaves deceitful. The
same occurs if one of the partitions AARB-delivers a value while the
other does not and reaches the timer (Lemma~\ref{lem:aarb-rec}).
\end{proof}
\begin{lemma}[AARB-Receive]
  \label{lem:aarb-rec}
  If a non-faulty process AARB-delivers $v$ from $p_s$, then all non-faulty processes eventually AARB-deliver $v$ from $p_s$.
\end{lemma}
\begin{proof}
  First, since $d+t<2h(d_r)-n$ it follows that deceitful and Byzantine processes can
not cause two non-faulty processes to AARB-deliver different
values (analogously to Lemma~\ref{lem:impagr}). Then, before a process $p_i$ AARB-delivers a value $v$, it
broadcasts a \READY message containing the certificate that justifies
delivering $v$. Thus, when $p_j$ receives that \READY message, it also
AARB-delivers v.
\end{proof}
 \begin{lemma}[AARB-Accountability]
   \label{lem:aarb-acc}
   If two non-faulty processes $p_i$ and $p_j$ AARB-deliver $v$ and $v'$,
respectively, such that $v\neq v'$, then all non-faulty processes
eventually receive PoFs of the deceitful behavior of at least $2h(d_r)-n$ processes (including $p_s$).

\end{lemma}
\begin{proof}
  Non-faulty processes broadcast the certificates of the values they
AARB-deliver, containing $h(d_r)$ signed \ECHO messages from distinct
processes. Therefore, analogous to Lemma~\ref{lem:impagr}, at least
$2h(d_r)-n$ processes must have sent conflicting $\ECHO$ messages, and
they will be caught upon cross-checking the conflicting
certificates. Also, some non-faulty processes must have received
conflicting signed $\INIT$ messages from $p_s$ in order to reach the
threshold $h(d_r)$ to AARB-deliver conflicting messages, meaning that
$p_s$ is also faulty.
\end{proof}
 \begin{lemma}[AARB-\Myproperty]
   \label{lem:aarb-aac}
   The \protocol's AARB protocol satisfies \myproperty.
\end{lemma}
\begin{proof}
  We prove here that if a number of faulty processes send conflicting messages to two
subsets $A,\,B\subseteq N$, each containing at least one non-faulty
process, then:
\begin{itemize}
\item eventually all non-faulty processes terminate without removing the faulty processes,
  or
\item eventually all non-faulty processes receive a PoF for these faulty processes and remove them from the committee, after which, if the source is non-faulty, they terminate.
\end{itemize}

W.l.o.g. we consider just $p_A\in A$ and $p_B\in B$. If they both
terminate despite the conflicting messages, we are finished. Suppose
instead a situation in which only one of them, for example $p_A$,
terminated AARB-delivering a value $v$. Then $p_A$ broadcast a \READY
message with enough $h(d_r)$ \ECHO messages in the certificate
$\ms{cert}$ for $p_B$ to also AARB-deliver $v$ and terminate. Let us
consider w.l.o.g. only one faulty process $p_i$. If a signature from
$p_i$ in $\ms{cert}$ conflicts with a local signature from $p_i$
stored by $p_B$, then $p_B$ constructs and broadcasts a PoF for $p_i$,
and then updates the committee and the
threshold. Then, it rechecks the certificate filtering out the
signature by $p_i$, which would cause $p_B$ to also AARB-deliver $v$
(since the threshold also decreased accordingly).

Suppose neither $p_A$ nor $p_B$ has terminated yet. Then, when the
timer is reached and they both broadcast the \INIT and \ECHO messages
they delivered, they will both be able to construct a PoF for $p_i$,
after which they update the committee and the threshold. Then, if the
source was non-faulty, non-faulty processes can terminate analogously
to the previous case.
\end{proof}

\subsubsection{\protocol binary consensus}
We focus in this section on the properties of \protocol's binary
consensus, AABC. We first prove that if all non-faulty processes
start a round $r$ with the same estimate $v$, then all non-faulty processes
decide $v$ in round $r$ or $r+1$. Then, we prove AABC-agreement in Lemma~\ref{lem:aabc-agr}, AABC-strong validity in Lemma~\ref{lem:aabc-strongva} and AABC-validity as Corollary~\ref{cor:aabc-val} of Lemma~\ref{lem:aabc-strongva}, AABC-\myproperty in Lemma~\ref{lem:aabc-aac}, AABC-termination in Lemma~\ref{lem:aabc-ter}, and AABC-accountability in Lemma~\ref{lem:aabc-acc}. This thus makes AABC the first \mypropertyadj binary consensus protocol, as we show in Theorem~\ref{thm:aabc-cons}.
\begin{lemma}
  \label{lem:aabc-aux}
  Assume that each non-faulty process begins round $r$ with the
estimate $v$. Then every non-faulty process decides $v$ either at the
end of round $r$ or round $r+1$.
\end{lemma}
\begin{proof}
By Lemma~\ref{lem:aabv-obl}, $v$ is eventually delivered to every
non-faulty process. By Lemma~\ref{lem:aabv-jus}, $v$ is the only value
delivered to each non-faulty process. As such, $v$ is the only value in
$\ms{bin\_vals}$ and the only value echoed by non-faulty processes,
since deceitful processes that prevent termination are removed from the
committee when the timer expires (and the threshold is updated). This
means that $v$ will be the only value in $\ms{vals}$. If $v=r\mod
2$ then all non-faulty processes decide $v$. Otherwise, by the same
argument every non-faulty process decides $v$ in round $r+1$.
\end{proof}
\begin{lemma}[AABC-Agreement]
  \label{lem:aabc-agr}
  If $d+t\leq 2h-n$, no two non-faulty processes decide different values.
\end{lemma}
\begin{proof}
  W.l.o.g. assume that the non-faulty process $p_i$ decides $v$ in round $r$. This
means that $p_i$ received $h(d_r)$ \ECHO messages in round $r$, and
that $vals=\{v\}$. Consider the \ECHO messages received by non-faulty
process $p_j$ in the same round. If $v$ is in $p_j$'s $vals$
then $p_j$ adopts estimate $v$ because $v=r \mod 2$. If instead
$p_j$'s $vals=\{w\},\,w\neq v$, then $p_j$ received $h(d_r)$ \ECHO messages
containing only $w$.

Analogously to Lemma~\ref{lem:impagr}, it is impossible for $p_j$ and for $p_i$ to receive
$h(d_r)$ \ECHO messages for $v$ and for $w$, respectively. We then conclude, by Lemma~\ref{lem:aabc-aux}, that every non-faulty
process decides value $v$ in either round $r+1$ or round $r+2$.
\end{proof}

\begin{lemma}[AABC-Strong Validity]
  \label{lem:aabc-strongva}
  If a non-faulty process decides $v$, then some non-faulty process proposed $v$.
\end{lemma}
\begin{proof}
  This proof is identical to Polygraph's proof of strong validity~\cite{civit2019techrep,civit2021}.
\end{proof}
\begin{corollary}[AABC-Validity]
  \label{cor:aabc-val}
  If all processes are non-faulty and begin with the same value, then
that is the only decision value.
\end{corollary}

\begin{lemma}[AABC-Termination]
  \label{lem:aabc-ter}
Every non-faulty process eventually decides on a value.
\end{lemma}
\begin{proof}
  This proof derives directly from Lemma~\ref{lem:aabc-aac}.
\end{proof}
\begin{lemma}[AABC-Accountability]
  \label{lem:aabc-acc}
  If two non-faulty processes output
disagreeing decision values, then all non-faulty processes eventually
identify at least $2h-n$ faulty processes responsible for that
disagreement.
\end{lemma}
\begin{proof}
  This proof is identical to Polygraph's proof of
accountability~\cite{civit2019techrep,civit2021}, with the a
generalization to any threshold $h(d_r)$ analogous to the one we make
in Lemma~\ref{lem:aarb-acc}.
\end{proof}

\begin{theorem}
  \label{thm:aabc-cons}
  \protocol's AABC solves the \mypropertyadj binary consensus problem.
\end{theorem}
\begin{proof}
   Corollary~\ref{cor:aabc-val} and
Lemmas~\ref{lem:aabc-agr},~\ref{lem:aabc-aac},~\ref{lem:aabc-ter},
and~\ref{lem:aabc-acc} prove AABC-validity, AABC-agreement,
AABC-\myproperty, AABC-termination and AABC-accountability,
respectively.
  \end{proof}
\subsubsection{General \protocol protocol}

We gather all the results together in this section, showing the proofs for the general \protocol protocol. We prove \myproperty in Lemma~\ref{lem:aac}, validity in Lemma~\ref{lem:val}, termination in Corollary~\ref{cor:ter}, agreement in Lemma~\ref{lem:agr}, and accountability in Lemma~\ref{lem:aac}. Finally, we prove that \protocol solves the \problem problem in Theorem~\ref{thm:consf}.

\begin{corollary}[Termination]
  \label{cor:ter}
  The \protocol protocol satisfies termination.
\end{corollary}
\begin{proof}
  Trivial from Lemma~\ref{lem:aac}.
  \end{proof}

\begin{lemma}[Validity]
  \label{lem:val}
  \protocol satisfies validity.
\end{lemma}
\begin{proof}
  This is trivial by Corollary~\ref{cor:aabc-val} and the proofs of
AARB. Suppose all processes begin \protocol with value $v$. If all
processes are non-faulty then every proposal AARB-delivered was
AARB-sent by a non-faulty process, and since all processes AARB-send
$v$, only $v$ is AARB-delivered.

Since initially processes only start an AABC instance for which they
can propose $1$, this means that eventually all processes start one
AABC instance proposing $1$. By Corollary~\ref{cor:aabc-val}, this
instance will terminate with all processes deciding $1$. Since the rest
of the AABC instances will eventually terminate by Lemma~\ref{lem:aabc-ter},
this means that processes will terminate at least one instance of AABC
outputting $1$. Upon calculating the minimum of all values (which are
all $v$) whose associated bit is set to $1$, all processes will decide
$v$.
\end{proof}

\begin{lemma}[Agreement]
  \label{lem:agr}
  The \protocol protocol satisfies agreement.
\end{lemma}
\begin{proof}
  The proof is immediate having Lemmas~\ref{lem:aabc-agr} and~\ref{lem:aarb-rec}.
\end{proof}

\begin{lemma}[Accountability]
  \label{lem:acc}
  If two non-faulty processes output
disagreeing decision values, then all non-faulty processes eventually
identify at least $2h-n$ faulty processes responsible for that
disagreement.
\end{lemma}
\begin{proof}
  The proof is immediate from Lemmas~\ref{lem:aabc-aac} and~\ref{lem:aarb-aac}.
\end{proof}

\begin{theorem}[Theorem~\ref{thm:consf1}]
  \label{thm:consf}
  The \protocol protocol with initial threshold $h_0\in(n/2,n]$ solves the
\problem problem if $d+t<2h_0-n$ and $q+t\leq n-h_0$.
\end{theorem}
\begin{proof}
  Corollary~\ref{cor:ter} and Lemmas~\ref{lem:aac},~\ref{lem:val},~\ref{lem:agr}, and~\ref{lem:aac} satisfy termination, \myproperty, validity,
agreement, and accountability, respectively.
%
\end{proof}
\begin{corollary}[Corollary~\ref{cor:consf1}]
  \label{cor:consf}
  The \protocol class of protocols solves the \problem problem if $n>3t+d+2q$.
\end{corollary}
\begin{proof}
  The proof is immediate from Theorem~\ref{thm:consf1} after removing $h_0$ from the system of two
inequations defined by $d+t<2h_0-n$ and $q+t\leq n-h_0$.
\end{proof}
\begin{corollary}[Corollary~\ref{cor:dec1}]
  The \protocol class of protocols solves $\Diamond$-consensus if $n>2t+d+q$.
\end{corollary}
\begin{proof}
    The proof is immediate from Theorem~\ref{thm:ec} after removing $h_0$ from the system of two inequations defined by $d+t<h_0$ and $q+t<n-h_0$.
\end{proof}
\subsubsection{Impossibility of consensus without \myproperty}In the proofs of Corollary~\ref{cor:imp} and Theorem~\ref{thm:imp} we
considered that deceitful faults do not prevent termination, that is,
that the protocol satisfies \myproperty. We show in
Corollary~\ref{cor:imp} the analogous result in the case where
deceitful processes can actually prevent termination, that is, if the
protocol does not satisfy \myproperty. In this case, since deceitful
can have the same impact as Byzantine (in that they can prevent either
agreement or termination), then the bounds decrease to
$n>3(t+d)+2q$. Note that other protocols that use authentication may
also be subject to this bound if they do not satisfy \myproperty, as
it is the case for Polygraph~\cite{civit2021}.

\begin{corollary}
  \label{cor:imp}
        It is impossible for a protocol that solves consensus without satisfying \myproperty to tolerate $t$ Byzantine, $d$ deceitful and $q$ benign processes if $n\leq 3(t+d)+2q$.
\end{corollary}
\begin{proof}
  The proof is analogous to Theorem~\ref{thm:imp} with the difference
that deceitful processes can actually prevent termination by sending
conflicting messages. Thus, we have $n+t+d\leq 2n-2q-2t-2d$,
which means $n>3(t+d)+2q$.
\end{proof}

\subsection{Extended complexities of \protocol}
\label{sec:extcomp}
\subsubsection{Complexities Before GST}
Before GST and in the presence of an adversary controlling $t$
Byzantine, $d$ deceitful, and $q$ benign processes, let $a$ be the number
of times the timer is reached before GST (i.e. $a\geq
\ceil{\frac{GST}{\Delta}}$), then the message and bit complexities of
AABC increase by a factor of $a\cdot n$, thus to $\mathcal{O}(an^3)$
and $\mathcal{O}(\lambda a n^4)$, respectively. The same occurs with
AARB's complexities. The time complexities are also affected by the
time $a$ to reach GST thus to $\mathcal{O}(an)$ for AABC and the general
\protocol, and $\mathcal{O}(a)$ for AARB.


Since there are $n$ pairs of reliable broadcasts and binary consensus
instances in the \protocol general protocol, 
the time complexity is
$\mathcal{O}(t+q+d)$, message complexity $\mathcal{O}(an^3)$ and bit
complexity $\mathcal{O}(\lambda a n^4)$. We show in
Table~\ref{tab:nogstcom} the worst-case complexities of the three protocols.

\begin{table}[htp]\centering
  \setlength{\tabcolsep}{11pt}
\begin{tabular}{lrrrr}
\toprule
  Complexity &
AARB &
AABC &
Basilic \\\midrule
  Time & $\mathcal{O}(a)$ & $\mathcal{O}(b)$ & $\mathcal{O}(b)$ \\
  Message & $\mathcal{O}(an^2)$ & $\mathcal{O}(an^3)$ & $\mathcal{O}(an^4)$ \\
  Bit& $\mathcal{O}(\lambda an^3)$& $\mathcal{O}(\lambda an^4)$& $\mathcal{O}(\lambda an^5)$\\
  \bottomrule
\end{tabular}

\caption{Time, message and bit complexities of Basilic's
AARB, AABC and the general \protocol protocol, before GST.}
\label{tab:nogstcom}
\end{table}
\subsubsection{Proofs} We prove in this section the complexities of \protocol, and of \protocol's AARB and AABC, which we presented in Section~\ref{sec:comps}.
\begin{lemma}[\protocol's AARB Complexity]
  \label{lem:compaarb}
  After GST and if the source is non-faulty, \protocol's AARB protocol has time complexity
$\mathcal{O}(1)$, message complexity $\mathcal{O}(n^2)$ and bit
complexity $\mathcal{O}(\lambda\cdot n^3)$.
\end{lemma}
\begin{proof}
  After GST, all non-faulty processes will have received a message
from each non-faulty process and from each deceitful processes by the time the
timer reaches $0$. Thus, either non-faulty processes can terminate, or they broadcast
their current list of \ECHO and \INIT messages, after which they
remove the detected deceitful processes, and they can terminate
too. Thus, the time complexity is $\mathcal{O}(1)$. Then, the message
complexity is $\mathcal{O}(n^2)$, as each non-faulty process
broadcasts at least one \ECHO and \READY message, and, in some
executions, a list of \ECHO messages that they delivered by the time
the timer reaches $0$. Since both this list and \READY messages
contain $\mathcal{O}(n)$ signatures, or $\mathcal{O}(\lambda n)$ bits,
the bit complexity of \protocol's AARB is $\mathcal{O}(\lambda n^3)$.
\end{proof}
\begin{lemma}[\protocol's AABC Complexity]
  \label{lem:compaabc}
  After GST, \protocol's AABC protocol has time complexity
$\mathcal{O}(n)$, message complexity $\mathcal{O}(n^3)$ and bit
complexity $\mathcal{O}(\lambda\cdot n^4)$.
\end{lemma}
\begin{proof}
  After GST, the \protocol protocol terminates in the first round (i)
whose leader is a non-faulty process and (ii) after having removed
enough deceitful faults so that they cannot prevent termination. Since
$t+d+q<n$, we have that (i) holds in $\mathcal{O}(n)$. As for every
added round in which deceitful faults prevent termination, a non-zero
number of deceitful faults are removed, we have that (ii) holds in
$\mathcal{O}(n)$ as well. This means that \protocol terminates in
$\mathcal{O}(n)$ rounds. In each round during phase $1$ of AABC,
non-faulty processes execute an ABV-broadcast of $\mathcal{O}(n^2)$,
obtaining $\mathcal{O}(n^3)$ messages. The bit complexity is
$\mathcal{O}(\lambda n^4)$ as each message may contain up to two
ledgers of $\mathcal{O}(n)$ signatures, or
$\mathcal{O}(\lambda n)$ bits. The complexities of phase 2 are
equivalent and obtained analogously to those of phase 1, as non-faulty
processes may broadcast $\mathcal{O}(n)$ signatures if
deceitful faults prevent termination of phase 2, or a
certificate if they decide in this round.
\end{proof}
\begin{theorem}
  \label{thm:compgen}
  The \protocol protocol has time complexity $\mathcal{O}(n)$, message complexity $\mathcal{O}(n^3)$ and bit complexity $\mathcal{O}(\lambda\cdot n^4)$.
\end{theorem}
\begin{proof}
  The proof is immediate from Lemma~\ref{lem:compaabc} and
Lemma~\ref{lem:compaarb} since \protocol executes $n$ instances of
AARB and after $n$ instances of AABC.
\end{proof}
 \addcontentsline{toc}{section}{Appendices}
\renewcommand{\thesubsection}{\Alph{subsection}}

\end{document}